%% file: unicorn.tex
\newif\ifdraft \draftfalse
\newif\iffull \fulltrue

\makeatletter \@input{tex.flags} \makeatother

\iffull
\documentclass[11pt]{article}
\else
\documentclass[prodmode,acmec]{ec-acmsmall} 
\fi

\usepackage{caption}
\usepackage{tabu}
\usepackage{booktabs}

\iffull
\let\chapter\section
\usepackage[ruled,noend]{algorithm2e}
\renewcommand{\algorithmcfname}{ALGORITHM}
\SetAlFnt{\small}
\SetAlCapFnt{\small}
\SetAlCapNameFnt{\small}
\SetAlCapHSkip{0pt}

\usepackage[margin=1.0in]{geometry}
\else

\usepackage[ruled,noend]{algorithm2e}
\renewcommand{\algorithmcfname}{ALGORITHM}
\SetAlFnt{\small}
\SetAlCapFnt{\small}
\SetAlCapNameFnt{\small}
\SetAlCapHSkip{0pt}
\IncMargin{-\parindent}

\acmVolume{X}
\acmNumber{X}
\acmArticle{X}
\acmYear{2015}
\acmMonth{2}
\fi

\iffull
\usepackage{amsmath, amssymb, amsthm}
\else
\usepackage{amsmath, amssymb}
\fi

\usepackage{verbatim}
\usepackage{color}
\usepackage{etex}
\usepackage{pgfplots}
\usepackage{esvect}
\usepackage{multirow}
\usepackage{multicol}

\definecolor{DarkGreen}{rgb}{0.1,0.5,0.1}
\definecolor{DarkRed}{rgb}{0.5,0.1,0.1}
\definecolor{DarkBlue}{rgb}{0.1,0.1,0.5}
\usepackage[]{hyperref}
\hypersetup{
    unicode=false,          
    pdftoolbar=true,        
    pdfmenubar=true,        
    pdffitwindow=false,      
    pdftitle={},    
    pdfauthor={}
    pdfsubject={},   
    pdfnewwindow=true,      
    pdfkeywords={keywords}, 
    colorlinks=true,       
    linkcolor=DarkRed,          
    citecolor=DarkGreen,        
    filecolor=DarkRed,      
    urlcolor=DarkBlue,          
}
\usepackage{xspace}
\usepackage{cleveref}
\usepackage{thmtools, thm-restate}

\iffull
\usepackage{natbib}
\else
\usepackage[numbers, sort&compress]{natbib}
\fi

\newcommand{\ar}[1]{\ifdraft \textcolor{green}{[Aaron: #1]} \fi}
\newcommand{\mk}[1]{\ifdraft \textcolor{red}{[Michael: #1]} \fi}
\newcommand{\rc}[1]{\ifdraft \textcolor{red}{[Rachel: #1]}\fi}
\newcommand{\sw}[1]{\ifdraft \textcolor{blue}{[Steven: #1]}\fi}

\newcommand\ABR{\vec{\text{BA}}}

\newcommand\RR{\mathbb{R}}
\newcommand\cA{\mathcal{A}}

\newcommand\cE{\mathcal{E}}

\newcommand\cM{\mathcal{M}}

\newcommand\cQ{\mathcal{Q}}
\newcommand\cR{\mathcal{R}}
\newcommand\cX{\mathcal{X}}

\newcommand\cV{\mathcal{V}}
\newcommand\cT{\mathcal{T}}

\newcommand{\br}{\vec{\text{BA}}}

\newcommand{\A}{\mathcal{A}}

\DeclareMathOperator{\polylog}{polylog}
\DeclareMathOperator{\Supp}{Supp}

\renewcommand{\tilde}{\widetilde}

\newcommand{\avgP}{\overline{p}}

\DeclareMathOperator*{\Expectation}{\mathbb{E}}

\newcommand{\vx}{\vec{x}}
\newcommand{\vp}{\vec{p}}

\newcommand{\PSL}{\mbox{{\sf PRESL}}\xspace}
\newcommand{\NPSL}{\mbox{{\sf NPRESL}}\xspace}
\newcommand{\PSN}{\mbox{{\sf PSummNash}}\xspace}
\newcommand{\SN}{\mbox{{\sf SummNash}}\xspace}
\newcommand{\SV}{\mbox{{\sf Sparse}}\xspace}
\newcommand{\EXP}{\mbox{{\sf EXP}}\xspace}
\newcommand{\DMW}{\mbox{{\sf DistMW}}\xspace}

\newcommand{\abr}{\text{-}\br}

\newcommand{\eps}{\varepsilon}
\def\epsilon{\varepsilon}

\DeclareMathOperator{\Lap}{Lap}
\DeclareMathOperator{\OPT}{OPT}
\renewcommand{\hat}{\widehat}

\iffull
\newcommand{\INDSTATE}[1][1]{\STATE\hspace{#1\algorithmicindent}}
\else
\newcommand{\INDSTATE}[1][1]{\hspace{#1\Indp}}
\fi

\iffull
\newtheorem{corollary}{Corollary}
\newtheorem{lemma}{Lemma}
\newtheorem{theorem}{Theorem}
\newtheorem{definition}{Definition}

\newtheorem{remark}{Remark}
\fi

\iffull
\else
\begin{document}
\markboth{Cummings et al.}{Privacy and Truthful Equilibrium Selection in Aggregative Games}
\fi

\title{Privacy and Truthful Equilibrium Selection for Aggregative Games}

\iffull
\author{
Rachel Cummings\footnotemark[1]
\quad
 Michael Kearns\footnotemark[2]
\quad
 Aaron Roth\footnotemark[2]
\quad
 Zhiwei Steven Wu\footnotemark[2]
}

\begin{document}

\date{\today}
\maketitle
\renewcommand{\thefootnote}{\fnsymbol{footnote}}
\footnotetext[1]{Computing and Mathematical Sciences, California Institute of Technology;
\texttt{rachelc@caltech.edu}. Research performed while the author was visiting the
University of Pennsylvania.
}
\footnotetext[2]{Computer and Information Science, University of Pennsylvania;
\texttt{\{mkearns,aaroth,wuzhiwei\}@cis.upenn.edu}
}
\renewcommand{\thefootnote}{\arabic{footnote}}
\else

\author{Rachel Cummings
\affil{California Institute of Technology}
Michael Kearns
\affil{University of Pennsylvania}
Aaron Roth
\affil{University of Pennsylvania}
Zhiwei Steven Wu
\affil{University of Pennsylvania}}
\fi

\begin{abstract}
  We study a very general class of games --- multi-dimensional aggregative
  games --- which in particular generalize both anonymous games and
  weighted congestion games. For any such game that is also \emph{large},
  we solve the equilibrium selection problem in a strong sense. In
  particular, we give an efficient {\em weak
    mediator\/}: a mechanism which has only the power to
  listen to reported types and provide non-binding suggested actions, such
  that (a) it is an asymptotic Nash equilibrium for every player to
  truthfully report their type to the mediator, and then follow its
  suggested action; and (b) that when players do so, they end up
  coordinating on a particular asymptotic pure strategy Nash equilibrium of
  the induced complete information game. In fact, truthful reporting is an
  \emph{ex-post} Nash equilibrium of the mediated game, so our solution
  applies even in settings of incomplete information, and even when player
  types are arbitrary or worst-case (i.e. not drawn from a common prior).
  We achieve this by giving an efficient differentially private algorithm
  for computing a Nash equilibrium in such games. The rates of convergence
  to equilibrium in all of our results are inverse polynomial in the number
  of players $n$. We also apply our main results to a multi-dimensional market
  game.

  Our results can be viewed as giving, for a rich class of games, a more
  robust version of the Revelation Principle, in that we work with weaker
  informational assumptions (no common prior), yet provide a stronger
  solution concept (ex-post Nash versus Bayes Nash equilibrium). In comparison to
  previous work, our main conceptual contribution is showing that weak
  mediators are a game theoretic object that exist in a wide variety of
  games -- previously, they were only known to exist in traffic routing
  games. We also give the first weak mediator that can implement an
  equilibrium optimizing a linear objective function, rather than
  implementing a possibly worst-case Nash equilibrium.
\end{abstract}

\iffull
\else
\category{F}{Theory of Computation}{Algorithmic Game Theory, Mechanism Design, Differential Privacy}

\terms{Algorithms, Truthfulness, Privacy}
\acmformat{Rachel Cummings, Michael Kearns, Aaron Roth, Zhiwei Steven Wu, 2015. Privacy and Truthful Equilibrium Selection for Aggregative Games.}
\fi
\maketitle

\iffull
 \vfill
 \thispagestyle{empty}
\setcounter{page}{0}
\pagebreak
\fi

\input{Introduction} 

\input{Preliminaries}

\input{HighDimension}

\input{Applications}

\input{OneDimension}

\section*{Future Work}

The most interesting open question in this line of work is whether
there exists a weak mediator that implements good behavior in
\emph{every} large game, where we only assume that
 the influence that any single player's action has on the utility of
others is diminishing with the number of players.  Recall that in
\citet{KPRU14}, it was shown that there exists a \emph{strong} mediator that
implements good behavior in any large game, by giving an algorithm
that privately computes a correlated equilibrium in any large game. An
equivalent result could be shown for weak mediators by giving an
algorithm that is able to compute (under the constraint of joint
differential privacy) a \emph{Nash} equilibrium, subject only to a
largeness condition on the game. Note that such an algorithm would not
be expected to be computationally efficient in general. However, at
the moment it remains open whether such an algorithm exists at all,
independent of efficiency concerns. Finally note that it might be
possible to construct weak mediators using tools other than
differential privacy -- there is no reason why such mediators could
not be deterministic. We do not at present have any other
similarly general tools for constructing these objects, but results
using other tools would be of significant interest.

\iffull
\bibliographystyle{plainnat}
\bibliography{unicorn.bbl}

\else
\bibliographystyle{acmsmall}
\bibliography{./refs}

\fi

\appendix

\input{Appendix}

\end{document}

%% file: Introduction.tex
\section{Introduction}\label{s.intro}

Games with a large number of players are almost always played, but
only sometimes modeled, in a setting of \emph{incomplete information}.
Consider, for example, the problem of selecting stocks for a 401k
portfolio among the companies listed in the S\&P500. Because stock
prices are the result of the aggregate decisions of millions of
investors, this is a large multi-player strategic interaction, but it
is so decentralized that it is implausible to analyze it in a complete
information setting (in which every player knows the types or
utilities of all of his opponents), or even in a Bayesian setting (in
which every agent shares common knowledge of a prior distribution from
which player types are drawn). How players will behave in such
interactions is unclear; even under settings of complete information,
there remains the potential problem of coordinating or selecting a
particular equilibrium among many.

One solution to this problem, recently proposed by~\citet{KPRU14}
and~\citet{RR14}, is to modify the game by introducing a \emph{weak
  mediator}, which essentially only has the power to listen and to
give advice. Players can ignore the mediator, and play in the original
game as they otherwise would have. Alternately, they can use the
mediator, in which case they can report their type to it (although
they have the freedom to lie). The mediator provides them with a
suggested action that they can play in the original game, but they
have the freedom to disregard the suggestion, or to use it in some
strategic way (not necessarily following it). The goal is to design a
mediator such that {\em good behavior\/} -- that is, deciding to use
the mediator, truthfully reporting one's type, and then faithfully
following the suggested action -- forms an ex-post Nash equilibrium in
the mediated game, and that the resulting play forms a Nash
equilibrium of the original {\em complete information\/} game, induced
by the actual (but unknown) player types. A way to approximately
achieve this goal -- which was shown in~\citet{KPRU14,RR14} -- is to
design a mediator which computes a Nash equilibrium of the game
defined by the reported player types under a stability constraint
known as {\em differential privacy\/} \citep{DMNS06}. Prior to our
work, this was only known to be possible in the special case of large,
unweighted congestion games.


In this paper, we extend this approach to a
much more general class of games known as {\em multi-dimensional
  aggregative games\/} (which among other things, generalize both anonymous games 
and weighted congestion games). In such a game, there is a vector of
linear functions of players' joint actions called an {\em
  aggregator\/}. Each player's utility is then a possibly {\em
  non-linear\/} function of the aggregator vector and their own
action. For example, in an investing game, the imbalance between buyers and
sellers of a stock, which is a linear function of actions, may be used in the
utility functions to compute prices, which are a non-linear function of the
imbalances (see Section~\ref{s.predict} for details). In an anonymous game,
the aggregator function represents the number of players playing each action.
In a weighted congestion game, the aggregator function represents the total
weight of players on each of the facilities. Our results apply to any
\emph{large} aggregative game, meaning that any player's unilateral change in
action can have at most a bounded influence on the utility of any other
player, and the bound on this influence should be a diminishing function in
the number of players in the game. Conceptually, our paper is the first to
show that \emph{weak mediators} are a game-theoretic object that exists in a
large, general class of games: previously, although defined, weak mediators
were only known to exist in traffic routing games \cite{RR14}.

This line of work can be viewed as giving robust versions of the
Revelation Principle, which can implement Nash equilibria of the
complete information game using a ``direct revelation mediator,'' but
without needing the existence of a prior type distribution. Compared
to the Revelation Principle, which generically requires such a
distribution and implements a Bayes Nash equilibrium, truth-telling
forms an ex-post Nash equilibrium in our setting. We include a
comparison to previous work in Table \ref{table:truthful}.

Finally, another important contribution of our work is that we are the first
to demonstrate the existence of weak mediators (in \emph{any} game) that have
the power to optimize over an arbitrary linear function of the actions, and
hence able to implement near optimal equilibria under such objective
functions, rather than just implementing worst-case Nash equilibria.

\begin{table}[ht]
\begin{center}
\begin{tabular}{ >{\centering\arraybackslash}m{2.2cm} >{\centering\arraybackslash}c >{\centering\arraybackslash}p{1.7cm} >{\centering\arraybackslash}p{1.8cm} >{\centering\arraybackslash}p{2.5cm} }
\toprule
  {\bf Mechanism \/} & {\bf Class of Games\/} & {\bf Common Prior? \/} & {\bf Mediator Strength\/} & {\bf Equilibrium Implemented\/} \\
\midrule
\shortstack{Revelation \\ Principle\\\citep{Mye81}} & Any Finite Game & Yes & Weak &  Bayes Nash \\[20pt]
\citet{KPRU14} & Any Large Game & No & Strong & Correlated \\[15pt]
\citet{RR14} & Large Congestion Games & No & Weak & Nash \\[15pt]
This Work & Aggregative Games & No & Weak & Nash \\[10pt]
\hline
\bottomrule
\end{tabular}
\caption{Summary of truthful mechanisms for various classes of games
  and solution concepts. Note that a ``weak'' mediator does not
  require the ability to verify player types. A ``strong'' mediator
  does. Weak mediators are preferred. }
\label{table:truthful}
\end{center}
\end{table}

\subsection{Our Results and Techniques}
Our main result is the existence of a mediator which makes truthful reporting
of one's type and faithful following of the suggested action (which we call
the ``good behavior'' strategy) an ex-post Nash equilibrium in the mediated
version of any aggregative game, thus implementing a Nash equilibrium of the
underlying game of complete information. Unlike the previous work in this
line~\citep{KPRU14, RR14}, we do not have to implement an arbitrary (possibly
worst-case) Nash equilibrium, but can implement a Nash equilibrium which
optimizes any linear objective (in the player's actions) of our choosing. We
here state our results under the assumption that any player's action has
influence bounded by $(1/n)$ on other's utility,
but our results hold more generally, parameterized by the ``largeness'' of
the game.

\begin{theorem}[Informal]\label{thm.informal}
  In a $d$-dimensional aggregative game of $n$ players and $m$
  actions, there exists a mediator that makes good behavior an
  $\eta$-approximate ex-post Nash equilibrium, and implements a Nash
  equilibrium of the underlying complete information game that
  optimizes any linear objective function to within $\eta$, where
\[
 \eta = O\left( \frac{\sqrt{d}}{n^{1/3}}\cdot \polylog(n, m, d)   \right).
\]
\end{theorem}

It is tempting to think that the fact that players only have small
influence on one another's utility function is sufficient to make any
algorithm that computes a Nash equilibrium of the game a suitable weak
mediator, but this is not so (see \citet{KPRU14} for an example). What
we need out of a mediator is that any single agent's report should
have little effect on the \emph{algorithm} computing the Nash
equilibrium, rather than on the payoffs of the other players.

The underlying tool that we use is differential privacy, which enforces the
stability condition we need on the equilibrium computation algorithm. Our
main technical contribution is designing a (jointly) differentially private
algorithm for computing approximate Nash equilibria in aggregative games. The
algorithm that we design runs in time polynomial in the number of players,
but exponential in the dimension of the aggregator function. We note that
since aggregative games generalize anonymous games, where the dimension of
the aggregator function is the number of actions in the anonymous game, this
essentially matches the best known running time for computing Nash equilibria
in anonymous games, even non-privately~\citep{DP08}. Computing exact Nash
equilibria in these games is known to be PPAD complete~\citep{CDO14}.

In the process of proving this result, we develop several techniques
which may be of independent interest. First, we give the first
algorithm for computing equilibria of multi-dimensional aggregative
games (efficient for constant dimensional games) even in the absence
of privacy constraints --- past work in this area has focused on the
single dimensional case \cite{KM02,B13}. Second, in order to implement
this algorithm privately, we develop the first technique for solving a
certain class of linear programs under the constraint of joint
differential privacy.

We also give similar results for a class of one-dimensional
aggregative games that permit a more general aggregation function and
rely on different techniques (\Cref{s.single}), and we show how
our main result can be applied to equilibrium selection in a
multi-commodity market (\Cref{s.predict}).

\subsection{Related Work}
\input{relwork}

%% file: relwork.tex
Conceptually, our work is related to the classic Revelation Principle
of~\citet{Mye81}, in that we seek to implement equilibrium
behavior in a game via a ``mediated'' direct revelation mechanism.
Our work is part of a
line, starting with~\citet{KPRU14} and continuing with~\citet{RR14},
that attempts to give a more robust reduction, without the need to
assume a prior on types. \citet{KPRU14} showed how to
privately compute correlated equilibria (and hence implement this
agenda) in arbitrary large games. The private computation of
correlated equilibrium turns out to give the desired reduction to a
direct revelation mechanism only when the mediator has the power to
verify types. \citet{RR14} rectified this deficiency by
privately computing Nash equilibria, but their result is limited to
large unweighted congestion games. In this paper, we substantially
generalize the class of games in which we can privately compute Nash
equilibria (and hence solve the equilibrium selection problem with a
direct-revelation mediator).

This line of work is also related to ``strategyproofness in the
large,'' introduced by \citet{AB12}, which has similar goals. In
comparison to this work, we do not require that player types be drawn
from a distribution over the type-space, do not require any smoothness
condition on the set of equilibria of the game, are algorithmically
constructive, and do not require our game to be nearly as
large. Generally, their results require the number of agents $n$ to be
larger than the size of the action set and the size of the type
set. In contrast, we only require $n$ to be as large as
the \emph{logarithm} of the number of actions, and require no bound at
all on the size of the type space (which can even be infinite).

Our work is also related to the literature on mediators in games
\citep{mediators1,mediators2}. In contrast to our main goal (which is
to implement solution concepts of the complete information game in
settings of incomplete information), this line of work aims to modify
the equilibrium structure of the complete information game. It does so
by introducing a mediator, which can coordinate agent actions if they
choose to opt in using the mediator. Mediators can be used to convert
Nash equilibria into dominant strategy equilibria \citep{mediators1},
or implement equilibrium that are robust to collusion
\citep{mediators2}. \citet{AMT09} considers mediators in
games of incomplete information, in which agents can misrepresent
their type to the mediators. Our notion of a mediator is related, but
our mediators require substantially less power than the ones from this
literature. For example, our mechanisms do not need the power to make
payments \citep{mediators1}, or the power to enforce suggested actions
\citep{mediators2}. Like the mediators of \citet{AMT09}, ours are
designed to work in settings of incomplete information and so do not
need the power to verify agent types --- but our mediators are weaker,
in that they can only make suggestions (i.e. players do not need to
cede control to our weak mediators).

The computation of equilibria in aggregative games (also known as
summarization games) was studied in \citet{KM02}, which gave efficient
algorithms and learning dynamics converging to equilibria in the
$1$-dimensional case. \citet{B13} also studies learning dynamics in
this class of games and shows that in the $1$-dimensional setting,
sequential best response dynamics converge quickly to equilibrium. Our
paper is the first to give algorithms for equilibrium computation in
the multi-dimensional setting, which generalizes many well studied
classes of games, including anonymous games. The running time of our
algorithm is polynomial in the number of players $n$ and exponential
in the dimension of the aggregation function $d$, which essentially
matches the best known running time for equilibrium computation in
anonymous games~\citep{DP08}.

We use a number of tools from differential privacy \cite{DMNS06}, as
well as develop some new ones. In particular, we use the advanced
composition theorem of \citet{boostingDP}, the exponential mechanism
from \citet{McSherryT07}, and the sparse vector technique introduced by
\citet{DNRRV09} (refined in \citet{HR10} and abstracted into its current
form in \citet{DR13}). We introduce a new technique for solving linear
programs under joint differential privacy, which extends a line of
work (solving linear programs under differential privacy) initiated by
\citet{HRRU14}.

Finally, our work relates to a long line of work initiated
by~\citet{McSherryT07} using differential privacy as a tool and
desideratum in mechanism design. In addition to works already cited,
this includes~\citet{NST12,NOS12,Xia13,GL13,CCKMV13,BMSS14,
schooltruth} among others. For a survey of this area see~\citet{PR13}.

%% file: Preliminaries.tex
\section{Model and Preliminaries}\label{s.prelim}

\subsection{Aggregative Games}\label{s.summgame}

Consider an $n$-player game with action set $\cA$ consisting of $m$
actions and a (possibly infinite) type space $\cT$ indexing utility
functions. Let $\vx = (x_i, \vx_{-i})$ denote a strategy profile in
which player $i$ plays action $x_i$ and the remaining players play
strategy profile $\vx_{-i}$. Each player $i$ has a utility function,
$u \colon \cT\times \cA^n \rightarrow [-1, 1]$, where a player with
type $t_i$ experiences utility $u(t_i, \vx)$ when players play
according to $\vx$. When it is clear from context, we will use
shorthand and write $u_i(\vx)$ to denote $u(t_i,\vx)$, the utility of
player $i$ at strategy profile $\vx$.

The utility functions in \emph{aggregative games} can be defined in
terms of a multi-dimensional \emph{aggregator} function $S\colon \A^n
\rightarrow [-W,W]^d$, which represents a compact ``sufficient
statistic'' to compute player utilities. In particular, each player's
utility function can be represented as a function only of her own
action $x_i$ and the aggregator of the strategy profile $\vx$:
$u_i(\vx) = u_i(x_i, S(\vx))$. We also assume $W$ to be polynomially
bounded by $n$ and $m$. In aggregative games, the function $S_k$ for
each coordinate $k\in [d]$, is an additively separable function: $
S_k(\vx) = \sum_{i=1}^n f_i^k(x_i). \footnote{In the economics
  literature, aggregative games have more restricted aggregator
  function: $S_k(\vx) = \sum_{i=1}^n x_i$. The games we study are more
  general, and sometimes referred to as \emph{generalized aggregative
    games}.} $

Similar to the setting of~\citet{KM02} and~\citet{B13}, we focus on
\emph{$\gamma$-aggregative games}, in which each player has a
\emph{bounded influence} on the aggregator:\[ \max_i \max_{x_i,
  x_i'\in \A}\left\|S(x_i,\vx_{-i}) - S(x_i',\vx_{-i}) \right\|_\infty
\leq \gamma, \mbox{ for all } \vx_{-i} \in \A^{n-1}.
\] That is, the greatest change a player can unilaterally cause to the
aggregator is bounded by $\gamma$. With our motivation to study large
games, we assume $\gamma$ diminishes with the population size $n$. We
also assume that all utility functions are $1$-Lipschitz with respect
to the aggregator: for all $x_i\in \A$, $\left|u_i(x_i, s) - u_i(x_i, s')
\right| \leq \|s - s'\|_\infty$.\footnote{Note that the influence that
  any single player's action has on the utility of others is also
  bounded by $\gamma$. If $\gamma=o(1/n)$, then any player's utility
  is essentially independent of other players' actions. Therefore, we
  further assume that $\gamma = \Omega(1/n)$ for the problem to be
  interesting. This will also simplify some statements.} 

For $\gamma$-aggregative games, we can express the aggregator more
explicitly as
\[S_k(\vec{x}) = \gamma \sum_{i=1}^n f^k_i(x_i),\] where
$f^k_i(x_i)$ is the influence of player $i$'s action $x_i$ on the
$k$-th aggregator function, and also $|f^k_i(x_i)|\leq 1$ for all
actions $i\in[n]$ and $x_i\in \cA$. Let $f^k_{ij} = f^k_i(a_j)$, where
$a_j$ denotes the $j$-th action in $\cA$.

We say that player $i$ is playing an $\eta$-\emph{best response} to
$\vx$ if $u_i(\vx) \geq u_i(x_i', \vx_{-i}) - \eta, \mbox{ for all }
x_i'\in \A$. A strategy profile $\vx$ is an $\eta$-\emph{pure strategy
  Nash equilibrium} if all players are playing an $\eta$-best response
in $\vx$. We also consider \emph{mixed strategies}, which are defined
by probability distributions over the action set. For any profile of
mixed strategies, given by a product distribution $\vp$, we can define
expected utility $u_i(\vp) = \Expectation_{\vx \sim \vp}{u_i(\vx)}$
and the expected aggregator
\begin{equation}
S_k(\vp) = \Expectation_{\vx\sim \vp}{S_k(\vx)} = \gamma \sum_{i=1}^n \sum_{j=1}^m f^k_{ij}\, p_{ij} =
\gamma\, \langle f^k , \vp\rangle \label{eq:expected-agg}.
\end{equation}
The \emph{support} of a mixed strategy $p$, denoted $\Supp(\vp_i)$, is
the set of actions that are played with non-zero probabilities. A
mixed strategy profile $\vp$ is a \emph{mixed strategy Nash
  equilibrium} if $u_i(\vp)\geq
\Expectation_{\vx_{-i}\sim\vp_{-i}}u_i(x_i', \vx_{-i})$ for all
$i\in[n]$ and $x_i'\in \cA$.

For each aggregator $s$, we define the \emph{aggregative best
  response}\footnote{Sometimes called \emph{best react}~\citep{B13}, and \emph{apparent
    best response}~\citep{KM02}.}
for player $i$ to $s$ as $\br_i(s) = \arg\max_{x_i\in A} \{u_i(x_i, s)\}$,
breaking ties arbitrarily. We define the \emph{$\eta$-aggregative best
  response set} for player $i$ to $s$ as 
\[
\eta\abr_i(s) = \{ x_i \in \A | u_i(x_i, s) \geq \max_{x_i'} u_i(x_i',
s) - \eta \}\] to be the set of all actions that are at most $\eta$
worse than player $i$'s exact aggregative best response.

\begin{remark}
Note that best response is played against the other players' actions
$x_{-i}$, but aggregative best response is played against the
aggregator value $s$. Aggregative best response ignores the effect of
the player's action on the aggregator, which is bounded by $\gamma$;
the player reasons about the utility of playing different actions as
if the aggregator value were \emph{promised} to be $s$. Nevertheless,
aggregative best response and best response can translate to each
other with only an additive loss of $\gamma$ in the approximation
factor. Furthermore, aggregative best responses to different
aggregators can translate to each other as long as the corresponding
aggregators are close. If $\|s-s'\|_\infty \leq \alpha$, then the
actions in $\eta\abr(s)$ are also in $(\eta+2\alpha)\abr(s')$. We
state these results more formally in the following lemmas.

\end{remark}

\begin{lemma}
\label{lem:nash-prop}
Let $\vx$ be a strategy profile such that player $i$'s action $x_i$ is
an $\eta$-best response. Then $x_i$ is also an $(\eta + \gamma)$-aggregative best response to $S(\vx)$.
\end{lemma}

\begin{proof}
  Let $s = S(\vx)$ and $s' = S(a, x_{-i})$ for some deviation $a\neq
  x_i$. Since $x_i$ is an $\eta$-best response, we know that $u_i(x_i,
  s) \geq u_i(a,s') - \eta$. By the bounded influence of player $i$,
  we know that $\|s - s'\|_\infty \leq \gamma$. Also, by the Lipschitz property of $u_i$, we have that $|u_i(a,s') -
  u_i(a,s)| \leq \gamma$. It follows that $u_i(a,s') \geq u_i(a,s)
  - \gamma$, and therefore $u({x_i}, s) \geq u_i(a,s) -\gamma
  -\eta$.
\end{proof}

\begin{lemma}\label{lem:abr-prop}
  Let $\vx$ be a strategy profile such that every player is playing
  $\eta$-aggregative best response to $S(\vx)$. Then we know that each
  player is playing $(\eta + \gamma)$-best response, and hence
  $\vx$ forms a $(\eta + \gamma)$-Nash equilibrium.
\end{lemma}

\begin{proof}
  Let $s = S(\vx)$ and $s' = S(a, x_{-i})$ for some deviation $a\neq
  x_i$. Since $x_i$ is an $\eta$-aggregative best response, we know
  $u(x_i, s) \geq u_i(a, s) - \eta$. We know that $\| s - s' \|_\infty \leq
  \gamma$ by bounded influence of $i$. Then by the Lipschitz property
  of $u_i$, $u_i(a, s) \geq u_i(a, s') - \gamma$. It follows that
  $u_i(x_i, s) \geq u_i(a, s') - \eta - \gamma$.
\end{proof}

\begin{lemma}\label{lem:abr-move}
  Suppose action $x_i$ is an $\eta$-aggregative best response to $s$
  for player $i$. Let $s'$ be aggregator such that $\|s - s'\|_\infty
  \leq \alpha$. Then $x_i$ is an $(\eta + 2\alpha)$-aggregative
  best response to $s'$.
\end{lemma}

\begin{proof}
Let $a \neq x_i$ be some deviation for player $i$. Since $x_i$ is an
$\eta$-aggregative best response to $s$, we have $u_i(x_i ,s) \geq
u_i(a, s) - \eta$. By the Lipschitz property of $u_i$, $u_i(x_i, s')
\geq u_i(x_i, s) - \alpha$ and also $u_i(a, s) \geq u_i(a, s') -
\alpha$. Combining these inequalities, we have $u_i(x_i, s') \geq
u_i(a, s') - \eta -2 \alpha$.
\end{proof}




\subsection{Mediated Games}
We now define games modified by the introduction of a mediator. A
mediator is an algorithm $M: (\mathcal{T}\cup \{\perp\})^n \to \cA^n$
which takes as input reported types (or $\perp$ for any player who
declines to use the mediator), and outputs a suggested action to each
player. Given an aggregative game $G$, we construct a new game $G_M$
induced by the mediator $M$.  Informally, in $G_M$, players have
several options: they can \emph{opt-out} of the mediator (i.e. report
$\perp$) and select an action independently of it. Alternately they
can \emph{opt-in} and report to it some type (not necessarily their
true type), and receive a suggested action $r_i$. They are free to
follow this suggestion or use it in some other way: they play an
action $f_i(r_i)$ for some arbitrary function $f_i:\cA\rightarrow
\cA$. Formally, the game $G_M$ has an action set $\cA_i$ for each
player $i$ defined as $\A_i = \A_i' \cup \A_i''$, where
$$
\cA'_{i} = \{(t_i, f_i): t_i \in \mathcal{T}, f_i : \cA \to \cA\}\quad \mbox{ and }
\A_i'' = \{ (\perp, f_i): f_i \text{ is constant} \} \mbox{.}
$$

Players' utilities in the mediated game are simply their expected utilities induced by the actions they play in the original game. Formally, they have utility functions $u_i'$:
$u_i'(t,f) = \mathbb{E}_{\vx \sim M(t)}[u_i(f(\vx))]$.
We are interested in finding mediators such that \emph{good behavior} is an ex-post Nash equilibrium in the mediated game. We first define an ex-post Nash equilibrium.

\begin{definition}[Ex-Post Nash Equilibrium]
A collection of strategies $\{\sigma_i:\cT\rightarrow \cA_i\}_{i=1}^n$
forms an $\eta$-approximate ex-post Nash equilibrium if for every type
vector $t \in \cT^n$, and for every player $i$ and action $x_i \in
\cA_i$:
$$u_i'(\sigma_i(t_i),\sigma_{-i}(t_{-i})) \geq
u_i'(x_i,\sigma_{-i}(t_{-i})) - \eta$$ That is, it forms an
$\eta$-approximate Nash equilibrium for \emph{every} possible vector
of types.
\end{definition} 
Note that ex-post Nash equilibrium is a very strong solution concept for incomplete information games because it does not require players to know a prior distribution over types.

In a mediated game, we would like players to truthfully report their type, and then faithfully follow the suggested action of the mediator. We call this \emph{good behavior}. Formally, the good behavior strategy is defined as $g_i(t_i) = (t_i, \mathrm{id})$ where $\mathrm{id}:\cA\rightarrow \cA$ is the identity function -- i.e.\ it truthfully reports a player's type to the mediator, and applies the identity function to its suggested action. 


In order to achieve this, we use the notion of \emph{joint
  differential privacy} defined in \citet{KPRU14} (adapted from
\emph{differential privacy}, defined in \citet{DMNS06} and presented here in \Cref{sec:privatetool}), as a privacy measure for mechanisms on agents' private
data (types). Intuitively, it guarantees that the output to all other
agents excluding player $i$ is insensitive to $i$'s private type, so
the mechanism protects $i$'s private information from arbitrary
coalitions of adversaries.

\begin{definition}[Joint Differential Privacy~\cite{KPRU14}]
  Two type profiles $t$ and $t'$ are $i$-neighbors if they differ only
  in the $i$-th component. An algorithm $\cM:\cT^n \rightarrow \A^n$
  is {\em $(\epsilon,\delta)$-joint differentially private} if for
  every $i$, for every pair of $i$-neighbors $t, t' \in \cT^n$, and for
  every subset of outputs $\mathcal{S} \subseteq \cA^{n-1}$,
\[
  \Pr[\cM(t)_{-i} \in \mathcal{S}] \leq \exp(\epsilon)\Pr[\cM(t')_{-i} \in \mathcal{S}] + \delta.
\]
If $\delta = 0$, we say that $\cM$ is {\em $\epsilon$-jointly differentially
  private}.
\end{definition}

We here quote a theorem of \citet{RR14}, inspired by \citet{KPRU14}
which motivates our study of private equilibrium computation.
\begin{theorem}[\cite{RR14,KPRU14}]
\label{thm:priv-media}
Let $M$ be a mechanism satisfying $(\epsilon,\delta)$-joint
differential privacy, that on any input type profile $t$ with
probability $1-\beta$ computes an $\alpha$-approximate pure strategy
Nash equilibrium of the complete information game $G(t)$ defined by type profile
$t$.  Then the ``good behavior'' strategy $g = (g_1,\ldots,g_n)$ forms
an $\eta$-approximate ex-post Nash equilibrium of the mediated game
$G_M$ for
$$\eta = \alpha + 2(2\epsilon + \beta +\delta) .$$
\end{theorem}

Our private equilibrium computation relies on two private algorithmic
tools, sparse vector mechanism (called $\SV$) and exponential
mechanism (called $\EXP$), which allows us to access agents' types in
a privacy-preserving manner.  (Full details in
\Cref{sec:privatetool}.)

%% file: HighDimension.tex
\section{Private Equilibrium Computation}\label{s.multi}
Let $G$ be a $d$-dimensional $\gamma$-aggregative game, and $L\colon
\cA^n \rightarrow \RR$ be a $\gamma$-Lipschitz\footnote{This can be
  achieved by scaling.} linear loss function:
\[
L(\vx) = \gamma \sum_i \ell_i(x_i) \quad \mbox{ and } \quad L(\vp) =
\gamma\Expectation_{\vx \sim \vp} L(\vx) =\gamma \sum_i \langle p_{ij}
, \ell_{ij}\rangle.
\]
where $0\leq \ell_{i}(a_j)\leq 1$ for all actions $a_j\in \A$, and
$\ell_{ij} = \ell_i(a_j)$.

Given any $\zeta \geq \gamma\sqrt{8n\log(2mn)}$, let $\cE(\zeta)$ be the set
of $\zeta$-approximate pure strategy Nash equilibria in the game
$G$,\footnote{We will show that $\cE(\zeta)$ is non-empty for $\zeta
  \geq \gamma\sqrt{8n\log(2mn)}$ in~\Cref{section:psl}.} and let
 \[\OPT(\zeta) = \min\{L(\vx) \mid \vx\in \cE(\zeta)\}.\]

We give the following main result:
\begin{theorem}
\label{main.thm}
For any $\zeta \geq \gamma\sqrt{8n\log(2mn)}$, there exists a mediator
$M$ that makes good behavior an $(\zeta + \eta)$-approximate ex-post
Nash equilibrium of the mediated game $G_M$, and implements an
approximate pure strategy Nash equilibrium $\vx$ of the underlying
complete information game with $L(\vx) \leq \OPT(\zeta) + \eta$, where
\[
\eta = O\left(n^{1/3} \gamma^{2/3} \sqrt{d}\cdot \polylog(n, m, d)
\right).\]
\end{theorem}

Recall that the quantity $\gamma$ is diminishing in $n$; whenever
$\gamma = O(1/n^{1/2 + \epsilon})$ for $\epsilon > 0$, the
approximation factor $\eta$ tends towards zero as $n$ grows large.
Plugging in $\gamma = 1/n$ and $\zeta = \gamma\sqrt{8n\log(2mn)}$
recovers the bound in Theorem \ref{thm.informal}.

This result follows from instantiating \Cref{thm:priv-media} with an
algorithm that computes an approximate equilibrium under joint
differential privacy, presented in \Cref{alg:pLP} as \PSL \ (Private
Equilibrium Selection).\footnote{We also present the full details of
  the non-private algorithm to compute equilibrium for aggregative
  games in \Cref{s.fixedpoint}.} We give here an informal description
of our algorithm, absent privacy concerns, and then describe how we
implement it privately, deferring the formal treatment to
\Cref{section:psl}. 

The main object of interest in our algorithm is the set-valued
function
\[\cV_\xi(\hat s) = \{S(\vp) \mid \text{for each }i,
\Supp(p_i)\subseteq \xi\abr_i(\hat s)\},\] which maps aggregator
values $\hat{s}$ to the set of aggregator values that arise when
players are randomizing between $\xi$-aggregative best responses to
$\hat{s}$. An approximate equilibrium will yield an aggregator
$\hat{s}$ such that $\hat{s} \in \mathcal{V}_\xi(\hat{s})$, so we wish
to find such a fixed point for $\cV_\xi$ (the value of $\xi$ will be
determined in the analysis, see~\Cref{section:psl}). Note that
\emph{pure strategy} Nash equilibria correspond to such fixed points,
but a-priori, it is not clear that fixed points of this function
(which may involve mixed strategies) are mixed strategy Nash
equilibria. This is because player utility functions need not be
linear in the aggregator, and so a best response to the expected value
of the aggregator need not be a best response to the corresponding
distribution over aggregators. However, as we will show, we can safely
round such fixed points to approximate pure strategy Nash equilibria,
because the aggregator will be well concentrated under rounding.

For every fixed value $\hat{s}$, the problem of determining whether
$\hat{s} \in \mathcal{V}_\xi(\hat{s})$ is a linear program (because
the aggregator is linear), and although $\Supp(p_i)\subseteq
\xi\abr_i(\hat s)$ is not a convex constraint in $\hat{s}$, the
aggregative best responses are fixed for each fixed value of
$\hat{s}$.  The first step of our algorithm simply searches through a
discretized grid of all possible aggregators $X=\{-W, -W + \alpha,
\ldots , W-\alpha\}^d$, and solves this linear program to check if
some point $\hat{s} \in \mathcal{V}_\xi(\hat{s})$. This results in a
set of aggregators $S$ that are induced by the approximate equilibria
of the game. Let $p_{ij}$ denote the probability that player $i$ plays
the $j$-th action. Then the linear program we need to solve is as follows:
\begin{equation}
\begin{aligned}
\label{eq:non-privateLP}
  &\forall k\in[d],\qquad \hat{s}_k - \alpha \leq \gamma\sum_{i=1}^n\sum_{j=1}^m
  f_{ij}^k p_{ij} \leq \hat{s}_k + \alpha\\
  &\forall i\in[n],\qquad \forall j\in \xi\text{-}\br_i(\hat{s}),    \qquad 0\leq p_{ij} \leq 1\\
  &\forall i\in[n],\qquad \forall j\notin \xi\text{-}\br_i(\hat{s}), \qquad
  p_{ij} = 0 
\end{aligned}
\end{equation}

Next, we need to find a particular equilibrium (an assignment of
actions to players) that optimizes our cost-objective function $L$.
This is again a linear program (since the objective function is
linear) for each $\hat{s}$. Hence, for each fixed point $\hat{s}\in
\cV_\xi(\hat s)$ we simply solve this linear program, and out of all
of the candidate equilibria, output the one with the lowest cost.
Finally, this results in mixed strategies for each of the players, and
we round this to a pure strategy Nash equilibrium by sampling from
each player's mixed strategy. This does not substantially harm the
quality of the equilibrium; because of the low sensitivity of the
aggregator, it is well concentrated around its expectation under this
rounding. The running time of this algorithm is dominated by the grid
search for the aggregator fixed point $\hat{s}$, which takes time
exponential in $d$. Solving each linear program can be done in time
polynomial in all of the game parameters.

Making this algorithm satisfy joint differential privacy is more
difficult. There are two main steps. The first is to identify
the fixed point $\hat{s} \in \mathcal{V}_\xi(\hat{s})$ that corresponds
the lowest cost equilibrium.  There are exponentially in $d$ many candidate aggregators to check, and with naive noise addition we would have to pay for this exponential factor in our accuracy bound. However, we take advantage of the
fact that we only need to \emph{output} a single aggregator -- the one
corresponding to the lowest objective value equilibrium -- and so
the \emph{sparse vector mechanism} $\SV$ (described in \Cref{s.sparse}) can be brought to bear, allowing us to pay only
linearly in $d$ in the accuracy bound.

The second step is more challenging, and requires a new technique: we
must actually solve the linear program corresponding to $\hat{s}$, and
output to each player the strategy they should play in equilibrium.
The output strategy profile must satisfy joint differential privacy.
To do this, we give a general method for solving a class of linear
programs (containing in particular, LPs of the form
\eqref{eq:non-privateLP}) under joint differential privacy, which may
be of independent interest. This algorithm, which we call $\DMW$
(described in \Cref{sec:distmw}), is a distributed version of the
classic multiplicative weights (MW) technique for solving LPs
\citep{mw-survey}. The algorithm can be analyzed by viewing each agent
as controlling the variables corresponding to their own mixed
strategies, and performing their multiplicative weights updates in
isolation (and ensuring that their mixed strategies always fall within
their best response set $\xi\text{-}\br_i(\hat{s}))$. At every round,
the algorithm aggregates the current solution maintained by each
player,
and then identifies a coordinate in which the constraints are far from being satisfied. The
algorithm uses the \emph{exponential mechanism} $\EXP$
(described in \Cref{s.exp}) to pick such a coordinate while
maintaining the privacy of the players' actions. The identification of such a
coordinate is sufficient for each player to update their own
variables. Privacy then follows by combining the privacy guarantee of
the exponential mechanism with a bound on the convergence time of the
multiplicative weights update rule. The fact that we can solve this LP
in a distributed manner to get joint differential privacy (rather than
standard differential privacy) crucially depends on the fact that the
sensitivity $\gamma$ of the aggregator is small. The algorithm \DMW
will find a set of strategies that approximately satisfy the linear
program  -- the violation on each coordinate is
bounded by
\[ E = {O}\left( \frac{n\gamma^2}{\eps} \polylog\left(n, m, d,
    \frac{1}{\beta}, \frac{1}{\delta} \right)\right)^{1/2}.
\]

\begin{algorithm}

\KwData{A type vector $t$, comparator parameter $\zeta$, linear cost
  function $L$, privacy parameters $(\eps, \delta)$, confidence
  parameter $\beta$}

\KwResult{An $\tilde{O}\left(\zeta + \frac{\left(\sqrt{n\eps} + d
    \right) \gamma}{\eps}\right)$-approximate pure strategy Nash
  equilibrium with cost objective no more than $\OPT(\zeta) +
  \tilde{O}\left(\frac{\left(\sqrt{n\eps} + d \right)
    \gamma}{\eps}\right)$}

\textbf{Initialize}: discretization resolution $\alpha =
      E_1 + E_2$, where
      \[
      E_1 = {\frac{100\gamma}{\eps} \left((d + 1)\log(2W)\log(n) +
          \log\left(\frac{6}{\beta} \right)\right)}, \phantom{m} \]\[E_2=
      100 \left(\frac{n\gamma^2}{\eps}\log\left(\frac{3d}{\beta}
        \right)\log(n)\sqrt{\log(m)\ln\left(\frac{1}{\delta}\right)}\right)^{1/2}
        \]

\tcc{Find a fixed point $\hat s \in \cV_\xi(\hat s)$ that corresponds
  to the lowest cost equilibrium}

 \Indp{\textbf{let} $\{a(\hat{s},
  \hat y)\} = \SV(t, \cQ, \alpha + E_1, 1, \eps)$}\tcp*{$\cQ$ formally
  defined in~\Cref{section:psl}.}

\eIf{all $a(\hat{s}, \hat y) = \perp$}{Abort\;}
{we have $(\hat{s}, \hat y)$ such that $a(\hat{s})
      \neq \perp$\;}
{\textbf{let} $\vec{p} = \DMW(LP(\hat
      s, \hat y), \epsilon, \delta,
      \alpha, \beta/3)$\;
\textbf{let} $\vx$ be an action profile sampled from the product
  distribution $\vp$\;}
\Indm{\textbf{Output: } $\vec{x}$\;}
\caption{Private Equilibrium Selection via LP: $\PSL(t, \zeta, L, \eps, \delta, \beta)$\label{alg:pLP}}
\end{algorithm}

The algorithm \PSL has the following guarantee:

\begin{restatable}{theorem}{multid}
\label{thm:multi-util}
Let $\zeta \geq \gamma\sqrt{8n\log(2mn)}, \eps,\delta,\beta\in (0,1)$.
$\PSL(t, \zeta, L, \eps, \delta, \beta)$ satisfies $(2\eps,\delta)$-joint
differential privacy, and, with probability at least $1 - \beta$,
computes a $(\zeta + 12 \alpha)$-approximate pure strategy equilibrium
$\vx$ such that $L(\vx) < \OPT(\zeta) + 5\alpha$, where
\[
\alpha = O\left( \frac{\left(\sqrt{n\eps} + d \right)\gamma}{\eps}\polylog\left(n,m,d, 1/\beta, 1/\delta \right) \right).
\]
\end{restatable}

We defer the full proof and technical details to \Cref{section:psl}.

\begin{remark}{The running time of this algorithm is exponential
  in $d$, the dimension of the aggregative game. For games of fixed
  dimension (where $d$ is constant), this yields a polynomial time
  algorithm. This exponential dependence on the dimension matches the
  best known running time for (non-privately) computing equilibrium in
  anonymous games by \cite{DP08}, which is a sub-class of aggregative
  games.}
\end{remark}

Theorem \ref{main.thm} then follows by instantiating Theorem
\ref{thm:priv-media} with $\PSL\left(t, \zeta, L, n^{1/3}\gamma^{2/3}
d^{1/2}, \frac{1}{n}, \frac{1}{n}\right)$ -- i.e. by setting $\epsilon
= n^{1/3} \gamma^{2/3} d^{1/2}$ and $\delta = \beta = \frac{1}{n}$.

%% file: Applications.tex
\section{An Application to Multi-Commodity Markets}\label{s.predict}

Here we give an application of our main result
to a natural market-based game, in which aggregator functions are used
to compute non-linear prices.

Consider a market with $d$ types of goods or contracts, which agents can
either buy or sell short (i.e. on each contract, an agent can be either
long, short, or neutral, and so we can think of actions as being
vectors $a \in \cA=\{-1,0,1\}^d$). In aggregate, the actions of all $n$
players will lead to a price for each contract, represented by a
vector $q \in [0,1]^d$. Agents have (potentially complicated) valuation
functions of their positions in the market, modeled as arbitrary
functions $v_i\colon \cA\rightarrow [-d,d]$, and utilities which are
quasilinear in money. Note that such valuation functions can model
arbitrary complementarity and substitute relationships between
contracts. As a result, the equilibria in this market can be complex
and diverse, and equilibrium selection becomes a problem.

Central to the game is a \emph{market maker} who
sets prices for each contract as a function of the demand. The precise
pricing rule that the market maker uses determines the structure of
the equilibria of the market. In both real markets and our idealized
game, one of the market maker's key objectives is to choose a pricing
rule that minimizes his \emph{worst-case loss} --- i.e. the loss he
might suffer over the buy and sell decisions of the market
participants (defined precisely below). A natural and realistic goal is
for this loss to be sublinear in the number $n$ of participants or
trades. For example, in random-walk models of price movements it is
typical for market maker loss to be on the order of $\sqrt{n}$ after
$n$ steps or trades.  Our model is
agnostic to the nature of the commodities being bought and sold ---
these could, for example, be contracts paying off as a function of the
realization of future events, making this a combinatorial prediction
market. See~\citet{MMfinance} and~\citet{MMprediction} for analyses of
market maker loss in both traditional finance and prediction market
models, respectively.

In the following, we show how to phrase the market
described here as an aggregative game.  We implement a market maker
that makes truthful reporting an approximate ex-post Nash equilibrium,
and computes an asymptotic Nash equilibrium of the underlying market,
all while guaranteeing that the market maker has loss bounded by
$O(n^{1/2+\epsilon})$ per commodity, for any constant $\epsilon > 0$
(i.e. almost achieving an overall loss $O(d\sqrt{n})$).

\subsection{Instantiation of the Market as an Aggregative Game}
We will formalize this setting as an $n$-player $\gamma$-aggregative game. The
action set of each player is $\A = \{-1,0,1\}^d$, where an action is a
$d$-dimensional vector of long, short, or neutral decisions (a
portfolio), where $1$ and $-1$ in the $k$-th coordinate respectively
indicate buying and selling a unit of the $k$-th security. Player
$i$'s private type is described by her private valuation function
$v_i: \cA \rightarrow [-d, d]$ that determines her value for any
portfolio of $d$ securities, held in positive or negative unit
quantities.

Given any strategy profile $\vx$, the \emph{imbalance} in each
security is the number agents buying minus the number of agents
selling: $I_k(\vx) =\sum_i (x_i)_k$. If the price of security $k$ is
$q_k$ and $a_j = 1$, then the player pays $q_k$; if $a_j = -1$, the
player is paid $q_k$; otherwise, the player receives no payment. The
price $q_k$ for each security is a (nonlinear) function of the
imbalance vector $I$ parameterized by $\lambda$:
\begin{equation}\label{eq:pricing}
q_k(I) =  \begin{cases}
    0 & \mbox{if } I_k < \frac{-\lambda}{2}\\
     {I_k}/{\lambda} + 1/2 &\mbox{if } \frac{-\lambda}{2} \leq I_k\leq \frac{\lambda}{2}\\
    1 & \mbox{if } I_k > \frac{\lambda}{2}
\end{cases}
\end{equation}
This simple ``hinge'' pricing rule is linear with slope $\lambda$ in a symmetric range of imbalances
around 0, and saturates at a price of 1 in the case of overdemand (too many buyers)
or 0 in the case of underdemand (too many sellers). We note that all the results
discussed here also hold for the standard exponential pricing rule often used
in prediction markets~\cite{MMprediction}.

\begin{figure}[h!]
\label{fig:hinge}
  \centering
    \includegraphics[width=0.5\textwidth]{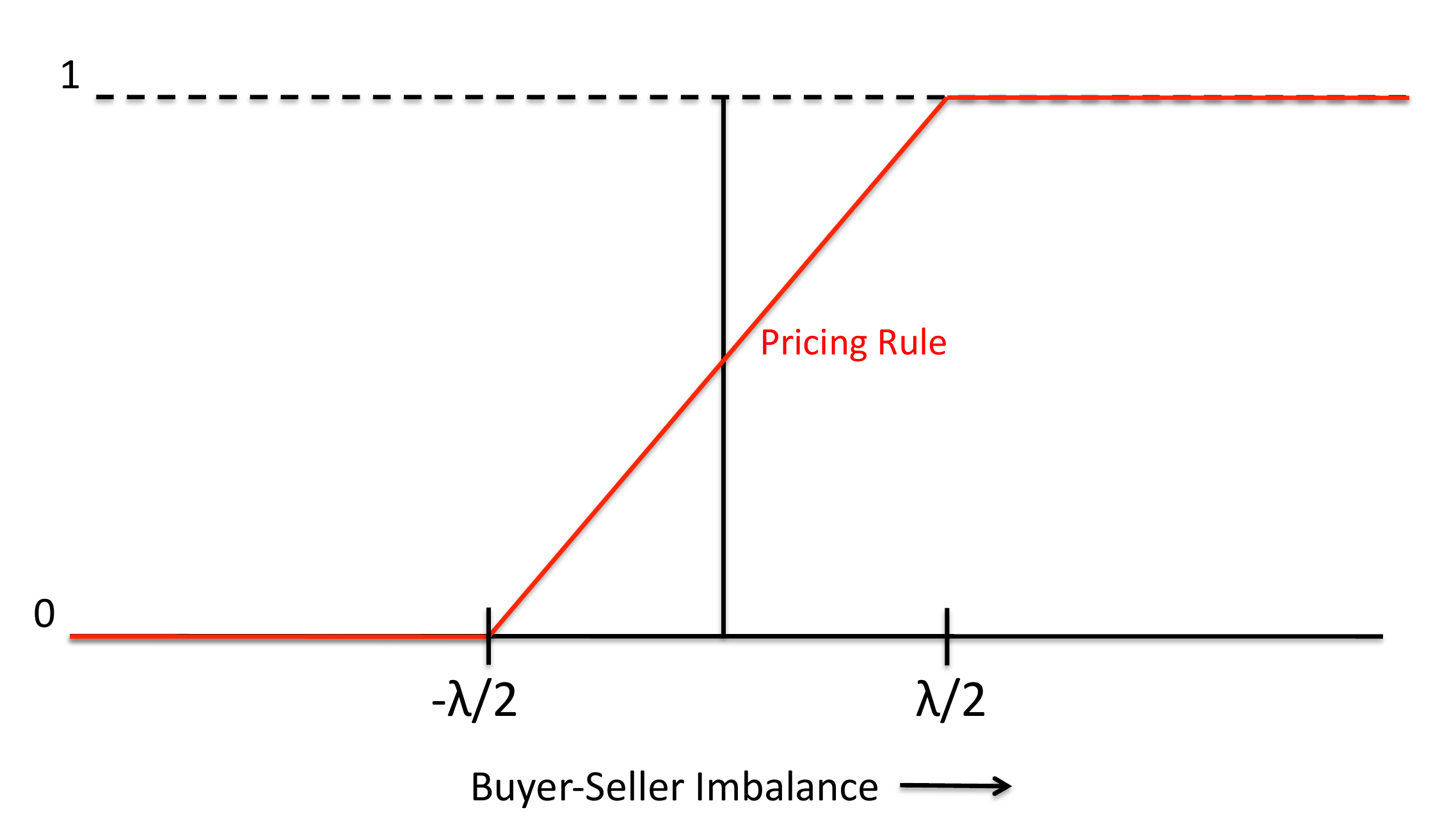}
  \caption{Hinge Pricing Rule}
\end{figure}



To apply our main result, we define our aggregator to be $S(\vx) =
I(\vx)/\lambda$. Conversely, when the aggregator has value $s$, the
imbalance vector is $\lambda s$ and each player $i$'s payoff function
(after rescaling) for playing action $x_i$ is
\[
u_i(x_i, s) = \frac{1}{2d}\left(v_i(x_i) - \langle x_i , q(\lambda s )
\rangle\right).\footnote{We normalize the utility function by $1/2d$
  to ensure that $u_i\in [-1, 1]$, so that it fits into the setting
  defined in~\Cref{s.summgame}.}
\]
For any fixed action $x_i$, the payoff $u_i$ is a
$1/2$-Lipschitz function of price vector $q$, and the price vector $q$
is an $(1/\lambda)$-Lipschitz function of the imbalance $I$, which in
turn is a $\lambda$-Lipschitz function of the aggregator. Therefore,
the payoff is a $1$-Lipschitz function of the aggregator.\footnote{In
  fact, the payoff is a $1/2$-Lipschitz function of the aggregator,
  which is a strictly stronger condition.} Here the range of the
aggregator is $[-n/\lambda, n/ \lambda]^d$, and each
player has bounded influence $\gamma = 1/\lambda$.

\subsection{Equilibrium Selection and  Market Maker's Loss}\label{s.pricing}


For each security $k$, if $I_k$ is positive, then there are $I_k$ more
buyers than sellers, and the market maker must sell to these players
at price $q_k(I)$.  The market maker will have to pay the maximum
price of $1$ to procure an extra copy of the item for each player in
the worst case, for a potential loss of $\ell_k = I_k
(1-q_k(I))$. Conversely, if $I_k$ is negative, then there are $ I_k$
more sellers than buyers, so the market maker must buy from these
players at price $q_k(I)$, for a potential loss of $\ell_k = -I_k\,
q_k(I)$. In total, the market maker's worst-case loss is $\sum_{k=1}^d
\ell_k$.

Now consider a mediated game in this market in which the market maker
wishes to elicit private valuation functions from all players and make
buy/sell recommendations to each player. We have the freedom to set
the pricing rule via choice of the parameter $\lambda$, but are also
faced with a bicriteria problem; we need to set prices to minimize the
potential loss of the market maker, while still incentivizing truthful
reporting from the players. Here, we demonstrate a trade-off between
incentives and market maker's potential loss. First, we have the
following lemma which bounds the market maker's loss as a function of
$\lambda$.

\begin{restatable}{lemma}{marketloss}
  The loss $\ell_k$ for the market maker in each security $k$ under
  the pricing rule defined in \Cref{eq:pricing} is bounded by
  $\lambda/16$.
\end{restatable}

\begin{proof}
Suppose $I_k > 0$. The loss
\begin{align*}
\ell_k = I_k (1- q_k(I)) &\leq I_k (1 - I_k/\lambda - 1/2)\\
&= I_k(1/2 -I_k/\lambda) \\
&= -1/\lambda \left(I_k - \lambda/4 \right)^2 + \lambda/16\leq \lambda/16.
\end{align*}
Suppose $I_k <0$, we also have
\begin{align*}
\ell_k = I_k (- q_k(I)) &\leq I_k (- I_k/\lambda - 1/2)\\
&= -1/\lambda \left(I_k + \lambda/4 \right)^2 + \lambda/16\leq \lambda/16.
\end{align*}
Thus, the loss is always bounded by $\lambda/16$.
\end{proof}

In order to guarantee sub-linear loss in each good, the market maker
needs to set $\lambda = o(n)$. Furthermore, since we have a
$(1/\lambda)$-aggregative game with each player's action set
consisting of $3^d$ actions, the market maker can use $\PSL$ as a
mediator to incentivize truthful reporting.
\begin{corollary}
  $\PSL$ as our market maker makes each player truthfully reporting their
  valuation function and following the market maker's recommendation
  form an $\eta$-approximate ex-post Nash equilibrium, where
\[
\eta = O\left( \sqrt{d}\left( \left(\frac{n}{\lambda^2}\right)^{1/3} +
\left(\frac{n}{\lambda^2}\right)^{1/2}\right)\cdot
  \polylog(n, d) \right).
\]
\end{corollary}


This result follows by instantiating Theorem
\ref{thm:priv-media} with $\zeta = \sqrt{8nd\log(3n)}/\lambda$.
For a fixed number of commodities $d$, we get asymptotic truthfulness
as long as the market maker sets $\lambda$ to be at least $n^{1/2+
  \eps}$ for any $\eps >0$. With this setting of $\lambda$, we also
guarantee that the market maker experiences worst-case loss at most
$O(n^{1/2 + \epsilon})$ per good.

%% file: OneDimension.tex
\section{Single Dimensional (Quasi)-Aggregative Games}\label{s.single}

In this section, we consider a more general class of games --
\emph{quasi-aggregative games}, in which the aggregator $S$ is not
required to have a linear structure as in aggregative games. We focus
on $\gamma$-quasi-aggregative games with a \emph{one-dimensional}
aggregator $S\colon \A^n \rightarrow [-W,W]$, and assume the same
properties of bounded influence and Lipschitz utilities.\footnote{This
  is identical to the setting in \citet{KM02} and \citet{B13}.} We
have the following result:

\begin{theorem}
\label{single.thm}
Let $G$ be a single dimensional $\gamma$-quasi-aggregative game for
some $\gamma < 1$. There
exists a mediator $M$ that makes good behavior an
$\eta$-approximate ex-post Nash equilibrium of the mediated game
$G_M$, and implements a Nash equilibrium $\vx$ of the underlying
complete information game, where
\[
\eta = O \left(\sqrt{\gamma}\cdot \polylog(n, m)
\right). 
\]
\end{theorem}

Similar to \Cref{s.multi}, our mediator is a jointly differentially
private algorithm that computes an approximate Nash equilibrium. The
algorithm is a private implementation of existing algorithms
\citep{KM02,B13}, so we use different techniques for these
single-dimensional games. Under certain assumptions, we can also
select equilibrium with respect to any Lipschitz objective function of
the aggregator (\Cref{s.select}).

\subsection{Private Equilibrium Computation}\label{s.privsummnash}

\mk{Maybe include figure from original KM paper illustrating main idea of non-private algo?}

Our algorithm $\PSN$, presented in \Cref{alg:psn}, is a privatized
version of the $\SN$ algorithm proposed in \citet{KM02}, that computes
an approximate Nash equilibrium under joint differential privacy.

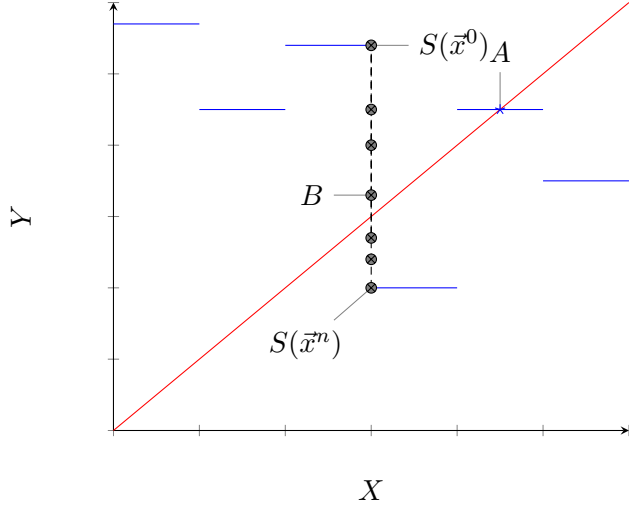
\begin{figure}
\centering
\begin{tikzpicture}
\begin{axis}[
    axis lines = left,
    xlabel = {$X$},
    ylabel = {$Y$},
    yticklabels={,,},
    xticklabels={,,}
]
\addplot [
  domain=-3:3, 
    samples=10, 
    color=red,
]
{x};

\addplot [
    domain=-3:-2, 
    color=blue,
    ]
    {2.7};
\addplot [
    domain=-2:-1, 
    color=blue,
    ]
    {1.5};
\addplot [
    domain=-1:0, 
    color=blue,
    ]
    {2.4};

 \addplot [
    domain=0:1, 
    color=blue,
    ]
    {-1};

    \addplot [
    domain=1:2, 
    color=blue,
    ]
    {1.5};

\addplot [
    domain=2:3, 
    color=blue,
    ]
    {0.5};

\addplot coordinates {
(0, -1)
(0, -.6)
(0,0.3)
(0,1)
(0, 2.4)
(0, 1.5)
(0, -0.3)
};

\addplot coordinates {
(1.5, 1.5)
  };
\node[coordinate,pin=above:{$A$}]
at (axis cs:1.5, 1.5) {};
\node[coordinate,pin=left:{$B$}]
at (axis cs:0, .3) {};
\node[coordinate,pin=right:{$S(\vx^0)$}]
at (axis cs:0,2.4) {};
\node[coordinate,pin=240:{$S(\vx^n)$}]
at (axis cs:0, -1) {};
\end{axis}
\end{tikzpicture}
\caption{A hypothetical plot of the function $V$. The ticks on the $X$
  axis correspond to the discretized points of aggregators $Z =\{-W,
  -W+\alpha,\ldots, W-\alpha\}$, and the height of the blue intervals
  correspond to the values of $V$ evaluated at these points (left
  endpoint of each interval). The diagonal line correspond to line
  $Y=X$.  The point labeled A is an example of a horizontal crossing
  the algorithm looks for in stage 1. The column of points including
  the point labeled B is an example of a vertical crossing the
  algorithm searches in stage 2. Each point in this column indicates a
  value of S realized on the sequence of strategy profiles $\vx^0$
  through $\vx^n$ defined in~\eqref{eq:walk}, while the point B itself
  is the value of S nearest the diagonal in this walk.}
\label{fig:sn}
\end{figure}

We will first briefly discuss the main idea of $\SN$, with reference
to~\Cref{fig:sn}. The main object of interest is the function $V$
defined on the aggregator space such that \[ V(s) = S(\ABR(s)),
\]
where $\ABR(s)$ denotes the aggregative best response profile to
aggregator $s$ (with each player breaking ties arbitrarily). The algorithm
will first discretize the aggregator space $[-W, W]$ into a discrete
set $Z = \{-W, -W+\alpha,\ldots , W-\alpha\}$, and evaluate $V$ on
each point of $Z$. In~\Cref{fig:sn}, the values of $V(s)$ for each
$s\in Z$ corresponds to the horizontal line segments.  Then the
algorithm finds an equilibrium in two stages.

In the first stage, the algorithm tries to find an approximate fixed
point of $V$ in the set $Z$.  Note that if we could identify an
aggregator $s$ such that $|V(s) - s| \leq \alpha$, then $\ABR(s)$
forms an $O(\alpha + \gamma)$-approximate equilibrium by
Lemma~\ref{lem:abr-prop} and Lemma~\ref{lem:abr-move}.  The algorithm
simply checks if there exist any $s\in Z$ that satisfy $|V(s) - s|
\leq \alpha$. In~\Cref{fig:sn}, the existence of an approximate fixed
point corresponds to a crossing point $A$, as the diagonal line
corresponds to the relation $Y = X$.  If such an $s$ is found, the
algorithm will simply suggest that each player $i$ play $\ABR_i(s)$
and halt.

Otherwise, the algorithm moves to the second stage, where it again
iterates over $Z$, this time to find two adjacent aggregators $s$ and
$s+\alpha$ such that $s > V(s) + \alpha$ and $s < V(s+\alpha) -
\alpha$. Such pair of $(s, s+\alpha)$ is guaranteed to exist because a
failure in the first stage implies that the two endpoints satisfy
$V(-W) > -W + \alpha$ and $V(W-\alpha) < W -
2\alpha$. 
Intuitively, the value of $V$ is ``too high'' at the lower endpoint,
and ``too low'' at the upper endpoint, so there must some ``crossing
point'' $s$ in the middle where $V(s)$ is close to $s$.

We can define a sequence of strategy profiles $\cX=\{\vx^0, \ldots,
\vx^n\}$, where
\begin{equation}
\vx^j_i = \begin{cases}
  \br_i(s) \mbox{ if }i\leq j\\
  \br_i(s + \alpha)\mbox{ otherwise}
      \end{cases}
\label{eq:walk}
\end{equation}
  Each profile in $\cX$ is a combination of some prefix in $\ABR(s)$
  and $\ABR(s+\alpha)$. The sequence of aggregators given by the
  profiles in $\cX$ is essentially a walk between $V(s)$ and
  $V(s+\alpha)$. By the assumption of bounded influence, changing the
  action of one player can change the aggregator value by at most
  $\gamma$, so every adjacent step in the walk has length no more than
  $\gamma$. Thus there must be an action profile $\vx\in \cX$ such
  $|V(x) - s| \leq \gamma$. Note that players' actions in $x$ come
  from both $\ABR(s)$ and $\ABR(s + \alpha)$, so all players are
  playing an $O(\alpha)$-aggregative best response to $s$ in $x$, and
  by Lemma \ref{lem:abr-prop}, $x$ forms an $O(\alpha +
  \gamma)$-approximate equilibrium.

The primary change we make to $\SN$ is that our algorithm needs to
access the players' aggregative best responses (defined by types) in a
privacy-preserving manner. Recall that there are three main parts that
require access to the players' private data:
\begin{enumerate}
\item for each $s\in Z$, check whether $V(s)$ is close to $s$;
\item for each $s\in Z$, check whether $s > V(s) + \alpha$ and also $s < V(s+\alpha) - \alpha$; and
\item for each $\vx \in \cX$, check whether $S(\vx)$ is close to a given $s\in Z$.
\end{enumerate}
To do that, we will first formulate these 3 conditions as 3 different
set of queries $\{Q_k\}, \{Q_k'\}$ and $\{Q_k''\}$ (detailed in
in~\Cref{alg:psn}), so that we can check these conditions by checking
whether the query values is below some threshold. Note that we only
need to identify at most one aggregator or strategy profile that
satisfies each condition even though we are answering a large
collection queries $(2W/\alpha + n)$. We take advantage of this fact
by using $\SV$, which gives a good accuracy guarantee, with error
scaling \emph{logarithmically} with the number of queries. This error
will eventually factor into the approximation factor of output Nash
equilibrium.

We here state the formal guarantee of the algorithm, and defer the
full proof and technical details to~\Cref{s.1dproofs}.

\begin{algorithm}
  \DontPrintSemicolon
\SetAlgoLined

   \caption{$\PSN(t, \eps,\alpha, \beta)$ \label{alg:psn}}
 \KwData{An
     $n$-player type vector $t$, privacy parameter $\eps$, accuracy
     parameter $\alpha$, and confidence parameter $\beta$}
 \KwResult{A $(10\alpha + 2\gamma )$-approximate
     Nash equilibrium}
\Indp{\textbf{Stage 1}\;}
{\textbf{for} each $-W/\alpha \leq k \leq W/\alpha - 1$, define a
  query $Q_k$ on the players' private payoff functions:}
      \[Q_k = | V(k\alpha) -  k\alpha |\]

{\textbf{let} $\{a_k\}=\SV(t, \{Q_k\}, 4\alpha, 1, \eps/3 )$\;}\tcp{try to find an approximate fixed point for $V$} 

\If{
  we have some $a_t \neq \perp$}{\textbf{Output}
  $\br(t\alpha)$}

{\textbf{Stage 2}\;}
{\textbf{for} each $-W/\alpha + 1\leq k\leq W/\alpha - 1$, define query
      $Q'_k$ on the players' private payoff functions:}
    \[ Q'_k = \max\left( \min(0, k\alpha - V((k -1 )\alpha)) ,
      -2\alpha \right) + \max\left( \min\left(0, V(k\alpha) - k\alpha
      \right), -3\alpha \right)
      \]
{\textbf{let} $\{a'_k\}=\SV(t, \{Q'_k\}, -4\alpha , 1, \eps/3 )$}

\eIf{ all  $a'_k = \perp$}{ Abort.}
{\textbf{let} $l$ be the index such that $a'_l  \neq \perp$\;}
\tcp{define a sequence of strategy profiles for the ``smooth walk''}
{\textbf{for} each $0\leq j\leq n$, \textbf{let}
  strategy profile $x^j$ be defined as}
\[
x^j_i = \begin{cases}
  \br_i(l\alpha) \mbox{ if }i\leq j\\
  \br_i((l-1)\alpha)\mbox{ otherwise}
\end{cases}
\]

{\textbf{let } query $Q_j'' = S(x^j)$\;} 
{\textbf{let} $\{a''_j\}=\SV(t, \{Q''_j\},
      \alpha + \gamma/2, 1, \eps/3 )$\;}
\eIf{some $a''_{j'} \neq \perp$}{
        \textbf{Output } $x^{j'}$}{Abort.}

\end{algorithm}

\begin{theorem}
\label{thm:psn-sum}
$\PSN(t, \eps, \alpha, \beta)$ satisfies $\eps$-joint differential
privacy, and with probability at least $1 - \beta$, computes a
$(10\alpha + 2\gamma)$-approximate pure strategy Nash equilibrium
as long as \[
\alpha \geq O\left( \frac{\gamma}{\eps}\polylog(n,m, 1/\beta) \right).
\]
\end{theorem}

\Cref{single.thm} then follows by instantiating Theorem
\ref{thm:priv-media} with $\PSN$ by setting $\eps = \sqrt{\gamma}$,
$\beta = 1/n$ and $\alpha = 100\sqrt{\gamma}(\log(12Wn^2)).$ (Recall
that $\gamma < 1$, so $\sqrt{\gamma}$ dominates $\gamma$).

%% file: Appendix.tex
\input{privacy}

\input{distMW}

\input{presl}

\input{nonprivLP}
\input{1dproofs}

%% file: privacy.tex
\section{Privacy Tools}\label{sec:privatetool}
We first state the formal definition of differential privacy
\cite{DMNS06}, which is a measure of the privacy of computations on
databases. In our setting, a database $D \in \cT^n$ contains $n$
players' private types, which determine their utility functions. Two
databases are \emph{neighboring} if they differ only in a single
entry.

\begin{definition}[Differential Privacy~\cite{DMNS06}]
  An algorithm $\cM:\cT^n\rightarrow \cR$ is {\em
$(\epsilon,\delta)$-differentially private} if for every pair of
neighboring databases $D, D' \in \cT^n$ and for every subset of possible
outputs $\mathcal{S} \subseteq \cR$,
\[ \Pr[\cM(D) \in \mathcal{S}] \leq \exp(\epsilon)\Pr[\cM(D') \in \mathcal{S}] + \delta.
\] If $\delta = 0$, we say that $\cM$ is {\em
$\epsilon$-differentially private}.
\end{definition}

We will make use of the following composition theorem, which shows how
the privacy parameters $\epsilon$ and $\delta$ ``compose'' nicely.


\begin{theorem}[Adaptive Composition~\cite{boostingDP}]
\label{thm:composition}
Let $\cM \colon \cT^n \rightarrow \mathcal{R}^T$ be a $T$-fold adaptive
composition\footnote{For a more detailed discussion of $T$-fold adaptive composition, see \citet{boostingDP}.} of $(\eps, \delta)$-differentially private mechanisms.
Then $\cM$ satisfies $(\eps', T\delta +\delta')$-differential privacy
for
\[ \eps' = \eps\sqrt{2T\ln(1/\delta')} + T\eps(e^{\eps} - 1).
\] In particular, for any $\eps\leq 1$, if $\cM$ is a $T$-fold
adaptive composition of $(\eps/\sqrt{8T\ln(1/\delta)},
0)$-differentially private mechanisms, then $\cM$ satisfies $(\eps,
\delta)$-differential privacy.
\end{theorem}

In the remainder of this section, we review two tools from the differential privacy
literature, namely the Sparse Vector Mechanism and the Exponential
Mechanism. Both will be used in our algorithms.

\subsection{Sparse Vector Mechanism}\label{s.sparse}

Our main tool from differential privacy is a slight modification of
the sparse vector mechanism from \citet{DNRRV09} (we follow the
presentation of \citet{DR13}). The sparse vector mechanism $\SV$ takes in a
sequence of low-sensitivity queries $\{Q_t\}$ on database $D$, and a
threshold $T$. The mechanism only outputs answers to those queries
with (noisy) answers below the (noisy) threshold,\footnote{The Sparse
  Vector Mechanism as presented in \citet{DR13} only answers queries
  with answers above a certain threshold. For the purposes of this
  paper, we use it instead to answer queries with answers below
  threshold. This modification does not change the analysis.} and
reports that all other queries were above threshold. There is also an
upper bound $c$ on the number of queries that can be answered. If more
than $c$ queries have answers below the threshold, the mechanism will
abort and not produce an output.

This mechanism is especially useful if an analyst is facing a stream
of queries and believes that only a small number of the queries will
have small answers. The Sparse Vector Mechanism allows the analyst to
identify and answer only the ``important'' queries, without having to
incur privacy cost proportional to all queries in the stream.

The \emph{sensitivity} of a query $Q$, denoted $\Delta(Q)$,
is an upper bound over all pairs of neighboring databases on the
amount that one entry can affect the answer to the query:
\[ \Delta(Q) = \max_{D, D' \mbox{ s.t. } |D \Delta D'|\leq 1} |Q(D) - Q(D')| \]
Note that because of our restriction to $\gamma$-aggregative games in Section \ref{s.summgame}, each
coordinate of the aggregative function is $\gamma$-sensitive.




\begin{algorithm}[h!]
\DontPrintSemicolon
\SetAlgoLined
 \caption{Sparse Vector Mechanism\quad
{$\SV(D, \{Q_t\}, T, c , \eps)$}   \label{alg:sparse}}

 \KwData{A private database $D$, an adaptively chosen stream of
   queries $\{Q_t\}$ of sensitivity $\gamma$, threshold $T$, total
   number of numeric answers $c$, and privacy parameter $\eps$}

 \KwResult{A stream of answers $\{a_t\} $}

\Indp{\textbf{Let} $\hat{T} = T + \Lap\left(\frac{2\gamma}{\eps} \right)$\;}
{\textbf{let} $\sigma = \frac{2c\gamma}{\eps}$\;}
{\textbf{let} $\text{count} = 0$\;}
\For{each query $Q_t$ on database $D$}
{\textbf{Let} $\nu_t = \Lap(\sigma)$ and $\hat{Q}_t =
      Q_t(D) + \nu_t$}
\eIf{ $\hat{Q}_t \leq \hat{T}$}{\textbf{Output} $a_t = \hat{Q}_t$\;
\textbf{Update} $\text{count} = \text{ count} + 1$\;
\If{$\text{count}\geq c $}{Abort}}
{\textbf{Output} $\perp$ }
\end{algorithm}

\begin{theorem}[\cite{DNRRV09}]\label{thm.sparse}
For any sequence of $N$ queries $Q_1, \ldots , Q_N$ such that $|\{k: Q_k(D) \leq T + \alpha | \leq c$, $\SV$ satisfies
$\eps$-differential privacy and, with probability at least $1 - \beta$, releases answers such
that for all $a_k \in \RR$,
\[
|a_k - Q_k(D)| \leq \alpha,
\]
and for all $a_k = \perp$,
\[
Q_k(D) \geq T - \alpha,
\]
where
\[
\alpha =
\frac{4c\gamma\left(\log{N} + \log(2c/\beta) \right)}{\eps}.
\]

\end{theorem}

\subsection{Exponential Mechanism}\label{s.exp}
The exponential mechanism~\citep{McSherryT07} is a powerful private
mechanism for selecting approximately the best outcome from a set of
alternatives, where the quality of an outcome is measured by a score
function relating each alternative to the underlying data. Let
$\cT^n$ be the domain of input databases, and $\mathcal{R}$ be
the set of possible outcomes, then a score function $q\colon
\cT^n \times \mathcal{R} \rightarrow \mathbb{R}$ maps each 
database and outcome pair to a real-valued score. The exponential
mechanism $\EXP$ 
instantiated with database $D$, a score function $q$, and a privacy
parameter $\eps$ is defined as
\[
\EXP(D, q, \eps) = \mbox{output } r \mbox{ with probability proportional
to } \exp\left(\frac{\eps q(D, r)}{2\Delta(q)}\right),
\]
where $\Delta(q)$ is the global sensitivity of score function $q$
defined as
\[
\Delta(q) = \max_{r, D, D' \mbox{ s.t. } |D \Delta D'|\leq 1} |q(D, r) - q(D',
r)|.
\]
Then exponential mechanism has the following property:
\begin{theorem}[\cite{McSherryT07}]
\label{thm:em}
$\EXP(D, q, \eps)$ satisfies $\eps$-differential privacy and, with
probability at least $1-\beta$, outputs an outcome $r$ such that
\[
q(D, r) \geq \max_{r'} q(D, r') - \frac{2\Delta(q)  \left(
  \log(|\mathcal{R}| / \beta)\right)}{\eps}.
\]
\end{theorem}

\subsection{Billboard Model}

In order to prove that our algorithms satisfy joint differential
privacy, we rely on a basic but useful framework -- the
\emph{billboard model}. Algorithms in the billboard model compute some
differentially private signal (which can be viewed as being visible on
a public billboard); then the output given to each player $i$ is
computed as a function only of this private signal, and the private
data of agent $i$. The following lemma shows that algorithms operating
in the billboard model satisfy joint differential privacy.

\begin{lemma}[Billboard Lemma. \cite{private-matching}]
  \label{lem:billboard}
Suppose $\cM\colon \cT^n \rightarrow \cR$ is $(\eps, \delta)$-differentially
private. Consider any set of functions $F_i\colon \cT_i\times \cR
\rightarrow \cR'$, where $\cT_i$ is the $i$-th entry of the input data.
The composition $\{F_i(\Pi_iD, \cM(D))\}$ is $(\eps,
\delta)$-jointly differentially private, where $\Pi_i$ is the
projection to $i$'s data.
\end{lemma}

%% file: distMW.tex
\section{Distributed Multiplicative Weights Algorithm}
\label{sec:distmw}

\ar{Maybe a high level description of the non-private algorithm is in
  order first. As it is, it is not at all clear where this class of
  linear programs comes from or why we need to solve it.} \ar{Also, is
  it clearly defined what we mean by solving this LP under joint
  differential privacy? We mean that each agent $i$ gets the solution
  to all of the variables they control, but we have to make clear that
  \emph{all} of these variables are represented by the $i$'th
  component of the output.} \rc{I agree with Aaron. Because we
  switched the order of Sections 3.1 and 3.2, this is no longer well
  motivated in the context of our paper} 
In order to compute an
equilibrium privately, we need to solve the linear program
in~\eqref{eq:non-privateLP} under joint differential privacy. This LP
has some nice structural properties which allow this to be possible. In particular, the variables are
well partitioned among the $n$ players, such that each player independently
controls a set of variables that must form a probability distribution.
Each player also has a private restricted feasible set defined by her
type (she needs to play an approximate aggregative best response to $\hat
s$ according to her private utility function). This motivates us to
solve the following more general linear program:
\begin{align}
&\forall k\in [d]\qquad  \gamma \sum_{i=1}^n \sum_{j=1}^m f^k_{ij} \,
p_{ij} = \gamma \langle f^k , p \rangle \leq
b_k \label{eq:public-con}\\
&\forall i\in[n]\qquad p_i =  (p_{i1},\ldots , p_{im}) \in R_i
\subseteq \left\{x\in \RR_{\geq 0}^m \mid \sum_j x_j = 1\right \}\label{eq:private-con}
\end{align}
where each $\left|f^k_{ij}\right|\leq 1$. In this LP, there are two
types of constraints. Each agent has a private constraint~\eqref{eq:private-con} for her own
variables, defined by the restricted feasible set $R_i$. We also have $d$ cross-agent
constraints~\eqref{eq:public-con} that require coordination among the
agents. Solving this LP under joint differential privacy guarantees
that the output variables to the other agents, $p_{-i}$, is insensitive
in agent $i$'s restricted set $R_i$. Our goal is to find a solution
$p$ that approximately satisfies all $d$ cross-agent
constraints~\eqref{eq:public-con} and exactly
satisfies all private constraints~\eqref{eq:private-con}.

Our algorithm \DMW is essentially a distributed version of
the multiplicative weights (MW) update algorithm~\citep{mw-survey}.

It proceeds in rounds and has each agent running an
instantiation of MW over her own private variables. At each round $t$,
the algorithm collects the variables from all players $p^t = (p^t_1,
\ldots , p^t_n) \in \mathbb{R}^{mn}$, and then selects an
approximately most violated cross-agent constraint under $p^t$ using the Exponential
Mechanism $\EXP$ (see \Cref{s.exp}), where the score for a constraint $\gamma
\langle f, p^t\rangle \leq b$ is defined as
\begin{equation}\label{eq:quality}
q(p^t, (f,b)) = \gamma \langle f, p^t\rangle - b.
\end{equation}
Note that each cross-agent constraint takes the same form, so
constraint $\gamma \langle f, p^t\rangle \leq b$ can be fully
described by the pair $(f, b)$. The mechanism then ``broadcasts'' the
selected constraint $(f, b)$, and each agent $i$ uses the $i$-th
segment of $f$: $f_i = (f_{i1}, \ldots , f_{im})$ as the loss vector
to update her instantiation of the MW distribution. After the
re-weighting update at each round, each player projects her vector of
variables into her private restricted set $R_i$, so the solution
always satisfies the private constraints. Finally, each agent takes
the average of the distributions from all rounds to get her output
distribution.

\begin{algorithm}[h!]
 \caption{Distributed Multiplicative Weights for Solving Linear
   Program:\qquad\qquad {\DMW$(FeasLP, \eps, \delta, \alpha,\beta)$}
 \label{alg:mw-lp} }

\KwData{ A feasibility LP $FeasLP$ of with 
     cross-agent constraints of the form \eqref{eq:public-con}, private
     constraints of the form \eqref{eq:private-con}, and quality score
   $q$ of the form \eqref{eq:quality}, privacy parameters $(\epsilon, \delta)$, accuracy parameter $\alpha$, and confidence parameter
   $\beta$}

\KwResult{ A solution $p$ that satisfies all private constraints and
  only violates any public constraint by at most $\alpha$}

{\textbf{Initialize} $p^1$ : $p_{ij}^1 = 1/m$ for all
     $i\in[n]$ and $j\in [m]$}

{\textbf{Let} $\displaystyle T =
     \frac{16 n^2 \gamma^2\log{m}}{\alpha^2}$ \qquad $\displaystyle\eps_0 =\frac{
       \eps}{ 2\sqrt{2T\ln(1/\delta)}}$ \qquad $\displaystyle\eta =
     \alpha/4n\gamma$}

{\textbf{For} each round $t\in\{1, \ldots,T\}$ }

\Indp{\textbf{Let} $(f^t,b^t) = \EXP(p^t, q, \eps_0)$}

{Each agent $i$ performs MW update: \textbf{for} each $j$}
\[
\hat{p}^{t+1}_{ij} = \exp(-\eta \cdot f^t_{ij}) \cdot p^t_{ij}
\]
\Indp{\textbf{Projection} with relative entropy:}
 \[
p_i^{t+1} = \arg\min_{x\in R_i} \mathbf{RE}(x|| \hat{p}_i^{t+1})
\]
\Indm{\textbf{Output} the average vector $\avgP = 1/T \sum_{t=1}^T p^t$}
\end{algorithm}

\begin{theorem}\label{thm.mwprivacy}
$\DMW(\cdot, \eps, \delta, \cdot, \cdot)$ satisfies $(\eps, \delta)$-joint differential privacy.
\end{theorem}

\begin{proof}
  Our algorithm works in the \emph{billboard model} introduced by
  \citep{private-matching}. In particular, the algorithm posts the
  violated constraint every round on a \emph{billboard} as a
  differentially private signal to all agents, such that every agent
  can see the signal and perform the MW update.

  The only sub-routines of \DMW that access the private data (i.e. private
  constraints) are the constraint selection at every round using the
  Exponential Mechanism. Thus, our mechanism has $T$ instantiations of an
  $(\eps_0,0)$-differentially private mechanism, where $\eps_0 =\eps/ (2\sqrt{2T\ln(1/\delta)})$. By the Adaptive
  Composition Theorem (\Cref{thm:composition} in
  \Cref{sec:privatetool}), we know that the selected constraints
  satisfy $(\eps, \delta)$-differential privacy. Note that the $i$-th
  component of the output is a function only of the selected constraints
  and the MW update rule. By \Cref{lem:billboard}, the algorithm
  satisfies $(\eps, \delta)$-joint differential privacy.
\end{proof}

\begin{theorem}\label{thm.mwacc}
  Suppose there is a feasible solution to $FeasLP$ with cross-agent
  constraints \eqref{eq:public-con} and private constraints
  \eqref{eq:private-con}.  Then with probability at least $1-\beta$,
  $\DMW(FeasLP,\eps, \delta, \alpha, \beta)$ outputs a
  solution $\avgP$ that satisfies all of the private constraints and
  $\alpha$-approximately satisfies all of the cross-agent constraints,
  for
\[ \alpha = {O}\left( \frac{n\gamma^2}{\eps} \polylog\left(n, m, d,
  \frac{1}{\beta}, \frac{1}{\delta} \right)\right)^{1/2}.
\]
\end{theorem}

\begin{proof}
Since each $|f_{ij}^k| \leq 1$,  the MW
algorithm gives a no-regret guarantee for each agent $i$:
\[
\frac{1}{T} \sum_t \langle f_i^t, p_i^t \rangle \leq \min_{p_i \in
  R_i}\frac{1}{T} \sum_t \langle f_i^t, p_i \rangle + \eta +
\frac{\log(m)}{T\eta} = \min_{p_i \in R_i}\frac{1}{T} \sum_t \langle
f_i^t, p_i \rangle + \frac{\alpha}{2n\gamma}
\]
for every agent $i$. Let $R = R_1\times \ldots \times R_n$, then the
joint play of all $n$ agents satisfy
\begin{equation}
\label{eq:mw-regret}
\frac{1}{T} \sum_t \left(\gamma \langle f^t , p^t \rangle - b^t\right)
\leq \min_{p \in R} \frac{1}{T} \sum_t \left(\gamma\langle f^t , p
  \rangle - b^t \right) + \alpha/2.
\end{equation}

Since there is feasible solution to the LP, we know that
\[
\min_{p  \in R} \frac{1}{T} \sum_t \left(\gamma\langle f^t , p \rangle - b^t \right) \leq 0,
\]
so from Equation \eqref{eq:mw-regret},
\begin{equation}\label{eq:bound1}
\frac{1}{T} \sum_t \left(\gamma \langle f^t , p^t \rangle - b^t\right)
\leq \alpha/2.
\end{equation}

By \Cref{thm:em}, with probability at least $1-\beta$, exponential
mechanism gives
\[
\sum_t \left(\gamma \langle f^t , p^t \rangle - b^t\right) \geq
\max_{(f, b)} \sum_t \left( \left(\gamma \langle f , p^t \rangle - b\right) -
\frac{2\gamma\log\left( \frac{dT}{\beta}\right)}{\eps_0} \right) \] and so,
\begin{equation}
\label{eq:em-regret}
\frac{1}{T} \sum_t \left(\gamma \langle f^t , p^t \rangle - b^t\right)
\geq \max_{(f, b)} \frac{1}{T} \left[ \sum_t \left(\gamma \langle f , p^t
  \rangle - b\right) \right] - \frac{2\gamma\log\left(
    \frac{dT}{\beta}\right)}{\eps_0} 
\end{equation}

Combining \Cref{eq:bound1,eq:em-regret} with the definition of $\avgP$, we get
\[
\max_{(f,b)} \left( \gamma\langle f, \avgP \rangle - b \right) = \max_{(f, b)} \frac{1}{T} \sum_t \left(\gamma \langle f , p^t
  \rangle - b\right) \leq
\frac{2\gamma\log\left( \frac{dT}{\beta}\right)}{\eps_0} + \alpha/2
\leq \alpha,
\]
as long as $\frac{2\gamma\log\left( \frac{dT}{\beta}\right)}{\eps_0}
\leq \alpha/2$. Plugging in for parameter $\eps_0$, this condition is
equivalent to
\begin{equation}\label{eq.alpha}
\alpha^2 \geq \frac{32\sqrt{2}n\gamma^2 \log\left( \frac{dT}{\beta}
  \right) \sqrt{\log{m}\ln\left( \frac{1}{\delta}\right)}}{\eps}.
\end{equation}

Plugging in $T$, we get our desired bound
\[ \alpha = \tilde{O}\left( \frac{n\gamma^2}{\eps} \log\left(\frac{d
n}{\beta}\right) \sqrt{\log(m) \ln(1/\delta)}\right)^{1/2}.
\]
\end{proof}

For simplicity, we set the target accuracy to be
\begin{equation}\label{eq:mw-acc}
\alpha =  100 \left(\frac{n\gamma^2}{\eps}\log\left(\frac{d}{\beta} \right)\log(n)\sqrt{\log(m)\ln(1/\delta)}\right)^{1/2}
\end{equation}
when calling $\DMW(\cdot, \eps, \delta, \alpha,\beta)$.\footnote{This
  accuracy level is achievable under mild conditions: as long as $\eps
  = O(1)$, $\gamma < 1$ and $1/\gamma$ is polynomially smaller than
  $n^{50}$ (we already assume $\gamma = \Omega(1/n)$), then the
  $\alpha$ in \Cref{eq:mw-acc} satisfies Inequality~\eqref{eq.alpha}.
}

%% file: presl.tex
\section{Details for $\PSL$}
\label{section:psl}

To recap, our goal is to select an approximate pure strategy
equilibrium of a $\gamma$-aggregative game such that the objective
cost function $L$ is approximately minimized. Before we begin, we need
to first define the benchmark $\OPT$ for the cost function, that is,
the set of equilibria to which we are comparing. We first state the
following result showing that approximate equilibrium always exists in
large games. It was observed by~\citet{LargeGame} and also proved
by~\citet{LipG} using a concentration argument, and we will state the
result in terms of aggregative games.

\begin{theorem}[\cite{LipG}]
Let $G$ be a $\gamma$-aggregative games with $n$ players and $m$
actions, then $G$ admits a $\left(\gamma
\sqrt{8n\log(2mn)}\right)$-approximate pure strategy Nash equilibrium.
\end{theorem}

Given any $\zeta \geq \gamma\sqrt{8n\log(2mn)}$, let $\cE(\zeta)$ be
the set of $\zeta$-approximate pure strategy Nash equilibria in the game.  We can then define the following benchmark
 \[\OPT(\zeta) = \min\{L(\vx) \mid \vx\in \cE(\zeta)\}.\]
Now we want to find an $(\zeta + \alpha)$-approximate pure strategy
equilibrium that achieves at most $\OPT(\zeta) + O(\alpha)$ using $\PSL$. The
algorithm has two stages. In the first stage, we try to identify an
aggregator $\hat s \in X$ such that there exists a mixed strategy
profile $\vp$ that satisfies three requirements:
\begin{enumerate}
\item all players are randomizing between approximate aggregative best
  responses to $\hat s$;
\item $S(\vp)$ is close to $\hat s$; 
\item and $L(\vp)$ is close to $\OPT(\zeta)$. 
\end{enumerate}
We find such an aggregator by solving linear programs based on the
discretized set of candidate aggregators and objective values.  In
particular, each linear program $LP(\hat s, \hat y)$ is defined by a
aggregator $\hat s \in X=\{-W, -W+\alpha, \ldots, W-\alpha\}^d$ and a
objective value for $L$: $\hat y\in \{0, \alpha, 2\alpha, \ldots,
n\gamma\}$. The sequence of queries we will feed to $\SV$ is the
objective values for these LP's:
\[
\cQ = \left\{ Q(\hat s, \hat y) \mid \hat s\in X, \hat y \in
  \{0, \alpha, \ldots, n\gamma\} \right\},
\]
\begin{align}\label{eq:lppriv}
 &\qquad \qquad Q(\hat{s}, \hat y) = \min a \\
 \mbox{ such that }\qquad &\forall k,\qquad
 \gamma\sum_{i=1}^n\sum_{j=1}^m
 f_{ij}^k p_{ij} \leq \hat{s}_k + a\label{eq:priv1}\\
 &\forall k,\qquad - \gamma\sum_{i=1}^n\sum_{j=1}^m
 f_{ij}^k p_{ij} \leq -\hat{s}_k + a\label{eq:priv2}\\
 &\;\; \qquad\qquad L(\vp) \leq \hat{y} + a \label{eq:priv6}\\
 &\forall i,\qquad \forall j\in \xi\text{-}\br_i(\hat{s}),    \qquad 0\leq p_{ij} \leq 1\label{eq:priv3}\\
 &\forall i,\qquad \forall j\notin \xi\text{-}\br_i(\hat{s}), \qquad  p_{ij} = 0\label{eq:priv4}\\
 &\forall i,\qquad \sum_{j=1}^m p_{ij} = 1\label{eq:priv5},
\end{align}
where $\xi = \gamma + \zeta + 2 \alpha$. 

We use $\SV$ to output the first $(\hat s, \hat y)$ such that the
approximate answer of $Q(\hat s, \hat y)$ is below the threshold
$\alpha + E_1$, where $E_1$ is the additive error bound for our
instantiation of $\SV$ (given by \Cref{thm.sparse}). This will
guarantee the actual value of $Q(\hat s, \hat y) \leq \alpha +
2E_1$.
%
%

During the second stage, the algorithm uses \DMW to compute a mixed
strategy profile $\vp$ by solving a modified version of the above LP,
denoted $LP(\hat{s},\hat{y})$, without the objective \eqref{eq:lppriv}
and with $a$ replaced by $\alpha + 2E_1$ in the constraints.  That is,
constraints \eqref{eq:priv1}-\eqref{eq:priv6} are further relaxed as
follows:
\begin{align}\label{eq:2ndstage}
  &\forall k,\qquad \gamma\sum_{i=1}^n\sum_{j=1}^m
  f_{ij}^k p_{ij} \leq \hat{s}_k + \alpha + 2E_1\\
  &\forall k,\qquad - \gamma\sum_{i=1}^n\sum_{j=1}^m
  f_{ij}^k p_{ij} \leq -\hat{s}_k + \alpha + 2E_1\\
  &\;\;\;\;\;\;\qquad L(\vp) \leq \hat{y} + \alpha + 2E_1.
\end{align}
Finally, we will output a pure strategy profile sampled from the
distribution defined by $\vp$.

\multid*

\begin{proof}
  Similar to the privacy proof in \Cref{thm.mwprivacy}, $\PSL$ also
  works in the Billboard model. Each player's action in the output is
  a function of only the ``broadcast'' information and her private type.
  Note that ``broadcast'' information comes from both $\SV$ and
  \DMW, which together satisfy $(2\eps,
  \delta)$-differential privacy. By the Billboard
  Lemma~\ref{lem:billboard}, our algorithm satisfies $(2\eps,
  \delta)$-joint differential privacy. 

  From \Cref{thm.sparse}, we know that the instantiation $\SV$ is
  $\eps$-differentially private, and with probability at least $1 -
  \beta/3$, has additive error bounded by
\[
err \leq \frac{4 \gamma \left((d+1)\log\left(\frac{2W}{\alpha} \right) +
    \log\left(\frac{6}{\beta} \right)\right)}{\eps} <  \frac{100 \gamma \left((d+1)\log(2W)\log(n) +
    \log\left(\frac{6}{\beta} \right)\right)}{\eps} = E_1 < \alpha,
\]
as long as $n^{25} > 1/\alpha$. For the rest of the proof, we
condition on this level of accuracy, which is the case except with
probability at most $\beta/3$. We know there exists an optimal
$\zeta$-approximate pure strategy equilibrium $\vx'$ such that
$L(\vx') = \OPT(\zeta)$. Note that each player's action $x'_i$ is
$\zeta$-best response, and is also $\xi$-aggregative best response
(recall $\xi = \gamma + \zeta + 2 \alpha$) to some discretized point
$\hat{s}\in X$, by Lemma~\ref{lem:nash-prop} and
Lemma~\ref{lem:abr-move}.

Since we set the threshold of $\SV$ to be $\alpha + E_1$, we are
guaranteed that $\SV$ will output a pair $(\hat{s}, \hat{y})$ such that
there exists some mixed strategy $\vp$ with
\[
\| S(\vp) - \hat s \|_\infty \leq \alpha + 2E_1 \qquad \mbox{ and }
\qquad L(\vp) \leq \hat{y} + \alpha + 2E_1 \leq \OPT(\zeta) +   2\alpha + 2E_1,
\]
and where every player only places weight on actions that are a
$\xi$-aggregative best response to $\hat{s}$.

Then $\vp$ is a feasible solution to the
second stage linear program, $LP(\hat s, \hat y)$. By \Cref{thm.mwacc}, \DMW will
output such a solution $\vp$ where
\[
\|S(\vp) - \hat{s}\|_\infty \leq \alpha + 2E_1 + E_2  \mbox{, and }
L(\vp) \leq \OPT(\zeta) + 2\alpha + 2E_1 + E_2, 
\]
except with probability $\beta/3$, where
\[
E_2= 100 \left(\frac{n\gamma^2}{\eps}\log\left(\frac{3d}{\beta}
  \right)\log(n)\sqrt{\log(m)\ln(1/\delta)}\right)^{1/2} < \alpha.
\]
Let $\vx$ be the action profile sampled from the mixed strategy $\vp$.
By McDiarmid's inequality, for each coordinate $k$:
\[
 \Pr\left[|S_k(\vx) - S_k(\vp)| \geq t\right] \leq
 2\exp\left(\frac{-2t^2}{n\gamma^2} \right) \mbox{, and }\Pr\left[|L(\vx) -
 L(\vp)| \geq t\right] \leq 2\exp\left(\frac{-2t^2}{n\gamma^2} \right).
\]
The union bound gives
\[
\Pr\left[\|S(\vx) - S(\vp)\|_\infty \geq t \mbox{ or } |L(\vx) - L(\vp)|
\geq t \right] \leq (d+1)2\exp\left(\frac{-2t^2}{n\gamma^2} \right).
\]
Then with probability at least $(1 - \beta/3)$, we can guarantee
\[
|L(\vx) - L(\vp)|, \|S(\vx) - S(\vp)\|_\infty \leq \left(\frac{n\gamma^2}{2}
  \ln\left(\frac{6d + 6}{\beta} \right)\right)^{1/2} \stackrel{\Delta}{=} E_3 < \alpha.
\]
Overall, we can guarantee the following with probability at least $1-\beta$
\[\|S(\vx) - \hat{s}\|_\infty \leq \alpha + 2E_1 + E_2 + E_3 < 4\alpha\]
\[
L(\vx) \leq L(\vp) + E_3 \leq \OPT(\zeta) + 2\alpha + 2E_1 + E_2 + E_3 < \OPT(\zeta)
+ 5\alpha.
\]
In $\vx$, all players are playing a $\xi$-aggregative best response to
$\hat{s}$, so by Lemma~\ref{lem:abr-move} they are also playing a
$(\xi + 8 \alpha)$-aggregative best response to $S(\vx)$. By
Lemma~\ref{lem:abr-prop}, all players in $\vx$ are playing a $(\xi +
\gamma + 8 \alpha)$-best response. Since $\gamma < \alpha$ and $\xi +
\gamma + 8 \alpha <\zeta + 12 \alpha$, then $\vx$ is an $\left(\zeta + O(\alpha)\right)$-pure
strategy Nash equilibrium.
\end{proof}

%% file: nonprivLP.tex
\section{Non-Private Equilibrium Computation via LP}\label{s.fixedpoint}

Here we show that similar techniques to those presented in
\Cref{section:psl} can be used to \emph{non-privately} compute and
select approximate equilibrium in $\gamma$-aggregative games. In this
setting, a better approximation factor is possible because we no
longer need to add noise to preserve privacy.

As before, the algorithm $\NPSL$ approximates the aggregator value
domain with a discretized grid of all possible aggregators $X=\{-W, -W
+ \alpha, \ldots , W-\alpha\}^d$, and consider a $\gamma$-Lipschitz
linear loss function $L\colon \A^n \rightarrow \RR$:
\[
L(\vx) = \gamma \sum_i \ell_i(x_i) \quad \mbox{ and } \quad L(\vp) =
\gamma\Expectation_{\vx \sim \vp} L(\vx) =\gamma \sum_i \langle p_{ij}
, \ell_{ij}\rangle.
\]
where $|\ell_{i}(a_j)|\leq 1$ for all actions $a_j\in \A$, and
$\ell_{ij} = \ell_i(a_j)$.



Let $\zeta \geq \gamma\sqrt{8n\log(2mn)}$, and define \[\OPT(\zeta) =
\min\{L(\vx) \mid \vx\in \cE(\zeta)\},\] where $\cE(\zeta)$ is the set
of $\zeta$-approximate pure strategy equilibria in the game.

 We first want to find a mixed strategy profile that could achieve the
 an objective value no more than $\OPT(\zeta)$.  This can be done by
 solving LP \eqref{eq:non-privatelp} for every $\hat s \in X$, denoted
 $LP(\hat{s})$, where $\xi = \zeta + \gamma + 2 \alpha$.
\begin{equation}
\begin{aligned}
\label{eq:non-privatelp}
  &\qquad \qquad \min_p L(p)\\
  &\forall k\in[d],\qquad \hat{s}_k - \alpha \leq \gamma\sum_{i=1}^n\sum_{j=1}^m
  f_{ij}^k p_{ij} \leq \hat{s}_k + \alpha\\
  &\forall i\in[n],\qquad \forall j\in \xi\text{-}\br_i(\hat{s}),    \qquad 0\leq p_{ij} \leq 1\\
  &\forall i\in[n],\qquad \forall j\notin \xi\text{-}\br_i(\hat{s}), \qquad p_{ij} = 0\\
  &\forall i\in[n],\qquad \sum_{j=1}^m p_{ij} = 1
\end{aligned}
\end{equation}
After solving $LP(\hat{s})$ for all $\hat{s} \in X$, algorithm \NPSL
selects the mixed strategy profile $\vp$ that gives the smallest
objective value among all the solutions to the LPs.  The algorithm
then rounds $\vp$ to get a pure strategy profile $\vx$.


\begin{algorithm}[H]
  \KwData{A type vector $t$, discretization parameter $\alpha$, and confidence parameter $\beta$}
  \KwResult{An $\tilde{O}\left(\alpha + \sqrt{n}\gamma \right)$-approximate pure strategy Nash equilibrium}
 
\Indp{\textbf{let} $s$ be the aggregator in $X$ that achieves the smallest objective value in $LP(\hat{s})$\;}
{\textbf{let} $\vec{p}$ be the solution to $LP(s)$\;}
{\textbf{let} $\vx$ be an action profile sampled from
      the product distribution $\vp$\;} 
{\textbf{Output}: $\vec{x}$}

  \caption{Non-Private Equilibrium Selection via LP: {$\NPSL(t, \alpha, \beta)$}  \label{alg:npLP}}
\end{algorithm}

\begin{theorem}
With probability at least $1- \beta$, \NPSL$(t, \alpha, \beta)$
computes a $\left(4\alpha + 2\gamma + 2E \right)$-approximate pure
strategy Nash equilibrium $\vx$ such that $L(\vx) \leq \OPT +E$, where
\[
E = O\left(\sqrt{n}\gamma \polylog(d, 1/\beta) \right).
\]
\end{theorem}

\begin{proof}
  Since the (unrounded) mixed strategy profile $\vp$ is a feasible
  solution to $LP(s)$, we know
\[
\| S(\vp) - s \|_\infty \leq \alpha.
\]

We also know that in $\vp$, every player is playing a
$\xi$-aggregative best response to $s$, and by
Lemma~\ref{lem:abr-move}, it follows that she is playing a $(\xi +
2\alpha)$-aggregative best response to $S(\vp)$.  

Let $\vx$ be the realized profile by sampling from $\vp$. By
McDiarmid's inequality, for each coordinate $k$:
\[
 \Pr\left[|S_k(\vx) - S_k(\vp)| \geq t\right] \leq
 2\exp\left(\frac{-2t^2}{n\gamma^2} \right) \mbox{, and }\Pr\left[|L(\vx) -
 L(\vp)| \geq t\right] \leq 2\exp\left(\frac{-2t^2}{n\gamma^2} \right).
\]
The union bound gives
\[
\Pr\left[\|S(\vx) - S(\vp)\|_\infty \geq t \mbox{ or } |L(\vx) - L(\vp)|
\geq t \right] \leq (d+1)2\exp\left(\frac{-2t^2}{n\gamma^2} \right).
\]
Then with probability at least $1 - \beta$, we can guarantee
\[
|L(\vx) - L(\vp)|, \|S(\vx) - S(\vp)\|_\infty \leq \left(\frac{n\gamma^2}{2}
  \ln\left(\frac{2d + 2}{\beta} \right)\right)^{1/2} \stackrel{\Delta}{=} E.
\]

Thus, with probability at least $1- \beta$, we know that in $\vx$,
each player is playing a $(\xi + 2\alpha + 2E)$-aggregative best
response to $S(\vx)$, and is therefore playing a $(\xi + 2\alpha + 2E
+\gamma)$-best response. Hence, we show $\vx$ forms an $(\zeta +
4\alpha+2\gamma+2E)$-approximate pure strategy equilibrium.

Note that the optimal $\zeta$-approximate pure strategy equilibrium
$\vx'$ with objective value $\OPT(\zeta)$ is also a feasible solution
to $LP(\hat{s})$ for some $\hat{s}\in X$, so we must have $L(\vp) \leq
\OPT(\zeta)$, which implies that $L(\vx) \leq \OPT(\zeta) + E$.
\end{proof}

%% file: 1dproofs.tex
\section{Details for Single Dimensional (Quasi)-Aggregative Games} \label{s.1dproofs}

\begin{proof}[Proof of \Cref{thm:psn-sum}]
  \Cref{alg:psn} only accesses the data through three instantiations
  of $\SV$, each of which satisfy $\epsilon/3$-differential
  privacy. By the Composition Theorem in \citet{DMNS06}, these three
  computations compose to satisfy $\eps$-differential privacy. Each
  player's action in the output strategy profile is a function only of $\SV$'s output and the player's private data (type). Thus, the
  algorithm works in the Billboard model, and by the Billboard
  Lemma~\ref{lem:billboard}, it satisfies $\eps$-joint differential
  privacy.





  We now prove that the algorithm computes an approximate Nash
  equilibrium. Let $err$ be the error in $\PSN$ due to its calls to
  $\SV$. We know by Theorem \ref{thm.sparse} that with probability at
  least $1 - \beta$, all three instantiations of $\SV$ have error at
  most
\[
err = \frac{100\gamma(\log(2Wn) + \log(6/\beta))}{\eps} \leq \alpha,
\]
by our assumption on $\alpha$.
For the rest of the proof we assume this level of accuracy, which is
the case except with probability $\beta$. 

First, consider the case that our algorithm outputs a strategy
profile in stage 1. We claim that this gives an $(10\alpha + \gamma
)$-approximate Nash equilibrium. Let $\br(k\alpha)$ be the output.
Then by the accuracy level of $\SV$,
\[
|V(k\alpha) - k\alpha| \leq 4\alpha + err \leq 5\alpha.
\]
Since each player's action is an aggregative best response to
aggregator value $k\alpha$, it is also a $10 \alpha$-aggregative best
response to $V(k\alpha)$ by Lemma~\ref{lem:abr-move}.
Thus, each player is playing a $(10 \alpha + \gamma )$-best response
as desired by Lemma~\ref{lem:abr-prop}.

Now suppose that the algorithm does not output anything in stage 1. We
argue that it will output a $(6\alpha + 2\gamma )$-approximate Nash
equilibrium in stage 2.

We first show that the algorithm's second $\SV$ outputs an index $l$
such that
\begin{equation}\label{eq:sandwich}
V((l - 1)\alpha) < l\alpha < V(l\alpha). 
\end{equation}
Since $\SV$ failed to output a strategy profile in stage 1, we know that 
\[
|V(k\alpha) - k\alpha| \geq 4\alpha - err \geq 3\alpha \mbox{ for all
} -W/\alpha\leq k \leq W/\alpha - 1.
\]
\rc{Verify that all this is still true.  I suspect it's
  not...}\sw{what's not true?}
Since $V(s)$ is in $[-W,W]$ for all $s$, we know that $V(-W) \geq
3\alpha-W$ and $V(W-\alpha) \leq W-4\alpha$.  Then there must exist an
index $l$ such that 
\[
V((l-1)\alpha) - l\alpha \geq 2\alpha \mbox{, and } l\alpha -
V(l\alpha) \geq 3\alpha,
\]
which implies that $Q'_l = -5\alpha \leq -4\alpha - err$. Such an $l$ satisfies \Cref{eq:sandwich} and will also be identified by $\SV$ as a below threshold query. Note that any
index $i$ which does not satisfy \Cref{eq:sandwich} must have $Q'_i\geq
-3\alpha \geq -4a + err$, so it will not be identified by the $\SV$ as a below
threshold query (since we set the threshold to be $-4\alpha$).

Finally, we claim that we can find an approximate equilibrium between
$\br((l-1)\alpha)$ and $\br(l\alpha)$. Let $x^j$ be the profile
defined in \Cref{alg:psn}. Then there exists $j'$ such
that 
\[
|S(x^{j'}) - l\alpha| \leq \gamma/2
\]
Suppose not. Since $S(x^0) > l\alpha > S(x^n)$, there exists $r$ such
that $S(x^r) > l\alpha >S(x^{r+1})$ such that
\[
S(x^r) - l\alpha > \gamma/2 \mbox{, and }  l\alpha - S(x^{r+1}) > \gamma/2.
\]
However, this violates our bounded influence assumption: $|S(x^{r+1}) -
S(x^r)| \leq \gamma$.

Since we set the threshold of the third $\SV$ to be $\alpha + \gamma/2$, we can find an index $j'$ such that
\[
|S(x^{j'}) - l\alpha| \leq 2\alpha + \gamma/2.
\]
Note that players in $x^{j'}$ are either playing an aggregative best
response to aggregative value $l\alpha$ or $(l-1)\alpha$. It suffices to bound the
payoff loss for the latter ones: \[
|S(x^{j'}) - l\alpha| \leq |S(x^{j'}) - (l-1)\alpha| \leq 3\alpha + \gamma/2.
\]
Thus an aggregative best response to $(l-1)\alpha$ remains a $(6\alpha
+ \gamma)$-aggregative best response to $S(x^{j'})$ by
Lemma~\ref{lem:abr-move}. By Lemma~\ref{lem:abr-prop}, each player can
only gain at most $\gamma $ by deviating from an apparent best
response, thus each player is playing $ (6\alpha + 2\gamma)$-best
response, which gives at least a $(10\alpha + 2\gamma )$-approximate
equilibrium.
\end{proof}

\subsection{Private Equilibrium Computation with a Lipschitz Objective}\label{s.select}


The $\PSN$ algorithm presented in Section~\ref{s.privsummnash} allowed
us to compute an approximate Nash equilibrium in any 1-dimensional
$\gamma$-quasi-aggregative game. However, if the game has multiple
approximate equilibria, it does not guarantee the quality of the
equilibrium we obtain. In this section we propose an algorithm to
select the approximate Nash equilibrium of the highest quality with
respect to a given objective. 
Our algorithm requires the following assumptions on the quasi-aggregative
game and objective:

\paragraph{Assumption 1}{Each player $i$ has a complete ordering $\succ_i$ over
  the action set $\A$, where $a \succ_i a'$ if and only if $S(a,
  x_{-i}) \geq S(a', x_{-i})$ for all $x_{-i} \in \A^{n-1}$. We say
  player $i$ is playing \emph{optimistically} if she is maximizing the
  aggregator value with her action, and playing \emph{pessimistically}
  if minimizing. }

\paragraph{Assumption 2}{Let $q\colon [-W, W] \rightarrow \RR$ be a
  score function that measures the quality of a aggregator value,
  where $q(s)$ is the quality of aggregator value $s$. We assume that
  $q$ is $\lambda$-Lipschitz in $s$.}
\newline

\citet{KM02} shows that any $\gamma$-quasi-aggregative
game admits a $4\gamma$-approximate pure Nash equilibrium. Given any
$\zeta \geq 4\gamma$, let $\cE(\zeta)$ be the set of
$\zeta$-approximate pure Nash equilibrium. Now define 
\[
\OPT(\zeta) = \{q(S(\vx)) \mid \vx \in \cE(\zeta)\}
\]
as our benchmark of the equilibrium quality. We will sometimes write
$\OPT$ for $\OPT(\zeta)$ when the context is clear.  We show that our
algorithm can compute an $(\zeta + O(\gamma))$-approximate equilibrium
with quality at least $\OPT + O(\lambda \gamma)$.

Similar to Algorithm \ref{alg:psn}, this algorithm also iterates
through all values $s$ in the discretized set $Z=\{-W, -W+\alpha,
\ldots, W-\alpha\}$, with players submitting their approximate
aggregative best response sets to $s$, $\xi\abr(s)_i$. Let $X(s) = \{
\vec{x} \mid \text{ each player is playing some } x_i \in \xi\abr(s)_i
\}$, where $\xi = 2\alpha + \gamma+ \zeta$. We are searching for an
approximate pure strategy Nash equilibrium $x' \in X(s)$ such that
$|S(x')- s| \leq \alpha$.  Note that the cardinality of $X(s)$ can
potentially be $\Omega(m^n)$, but by Assumption 1, we can
compute the upper and lower bound efficiently
\[S_{max}(s) = \max_{x\in X(s)} S(x) \mbox{, and } S_{min}(s) = \min_{x\in X(s)} S(x)\]
by asking players to play optimistically and pessimistically, respectively. 

If we find such an $s$ that has an approximate Nash equilibrium
strategy profile, there are three cases: either $S_{max}(s)$ or
$S_{min}(s)$ is close to $s$, or $s\in [S_{min}(s), S_{max}(s)]$. If
we are in the first two cases, the algorithm simply outputs the
corresponding to the optimistic or pessimistic strategy profile,
respectively. In the third case, the algorithm perform the same smooth
walk as in \Cref{alg:psn} from the optimistic profile to the
pessimistic profile, and outputs an intermediate profile $x'$ such
that $|S(x') - s|$ is small.

Since we are interested in computing the \emph{best} equilibrium, we
iterate through aggregators in order of their quality score. Let
$s_1\succ s_2 \succ \ldots \succ s_{2W/\alpha}$ be an ordering over
the set of discretized aggregator values $Z$, such that $q(s_{i}) \geq
q(s_{i+1})$. The algorithm will sequentially consider $s_i$ according
to this ordering $\succ$, to compute an approximate equilibrium with
aggregator value that maximizes $q$.
\begin{algorithm}
\DontPrintSemicolon
\KwData{An $n$-player type vector $t$, comparator equilibrium class parameter $\zeta$, privacy parameter $\eps$, accuracy parameter $\alpha$, and confidence parameter $\beta$}
\KwResult{$(10\alpha + 3 \gamma + \zeta)$-approximate Nash equilibrium with quality score at least $\OPT(\zeta) - 5\alpha\lambda$}
{\textbf{Initialize:} $G$ such that $q(G)=-W$\;} 
  \For{each aggregator value $s\in \{s_1, s_2,
    \ldots, s_{2W/\alpha}\}$} 
{\textbf{let} $X(s) = \{\vec{x} \mid \mbox{each player has } x_i \in (2\alpha  + \gamma + \zeta)\text{-}\br(s)_i \}$\;
\textbf{let} $S_{max}(s) = \max_{x\in X(s)} S(\vec{x})$
    and $S_{min}(s) = \min_{x\in X(s)} S(\vec{x})$\;}

\For{$1\leq k \leq 2W/\alpha$}{\textbf{let} queries
  \[
  Q_k = |S_{max}(s_k) - s_k| \qquad Q_k' = |S_{min}(s_k) - s_k|
  \]
  \[
  Q_k'' = \max\left(\min\left( S_{min}(s_k) - s_k, 0 \right),
    - 2\alpha \right) + \max\left(\min\left( s_k - S_{max}(s_k), 0 \right),
    - 2\alpha \right)
  \]}
{\textbf{let} $\{a_k\} = \SV(t, \{Q_k\},
    3\alpha, 1, \eps/4)$\;}
\If{some $a_i\neq \perp$}{\textbf{let} $y = x_{max}(s_i)$ and the associated aggregator $S' = s_i$}
{\textbf{let} $\{a'_k\} = \SV(t, \{Q'_k\},
    3\alpha, 1, \eps/4)$\;}
\If{some $a'_i\neq \perp$}{ \If{$s_i \succ S'$}{\textbf{ let } $y = x_{min}(s_i)$ and the associated aggregator $S' = s_i$}}

{\textbf{let} $\{a''_k\} = \SV(\{t, \{Q''_k\},
    3\alpha, 1, \eps/4)$\;}
\If{some $a''_l \neq \perp$}{
\For{each $0\leq j\leq n$}{ \textbf{let}
      strategy profile $x^j$ be defined as
    \[
    x^j_i = \begin{cases}
      x_{max}(s_l)_i \mbox{ if }i\leq j\\
      x_{min}(s_l)_i \mbox{ otherwise}
      \end{cases}
    \]}
      {\textbf{let } query $Q_j''' = S(x^j)$\;} 
      {\textbf{let} $\{a'''_j\}=\SV(t, \{Q_j'''\},
      \alpha + \gamma/2, 1, \eps/4)$\;}
    \If{some $a'''_{j'} \neq \perp$}{ \If{$s_l \succ S'$}{\textbf{let} $y = x^{j'}$  } }

\eIf{$y$ is defined}{\textbf{Output} $y$}{Abort}
}

  \caption{Private Equilibrium Selection with a Lipschitz
  Objective \label{alg:pselection}}
\end{algorithm}

\begin{theorem}
For any $\zeta\geq 4\gamma$, \Cref{alg:pselection} satisfies
$\eps$-joint differential privacy, and with probability at least $1 -
\beta$, outputs a $(10\alpha + 3 \gamma + \zeta)$-approximate Nash
equilibrium with quality score at least $\OPT(\zeta) -
5\alpha\lambda $, for any 
\[
\alpha \geq O\left(\frac{\gamma}{\eps}\polylog(n, m, 1/\beta) \right).
\]
\end{theorem}


\begin{proof}
  \Cref{alg:pselection} only accesses the data through four
  instantiations of $\SV$, each of which answers at most one query
  with $\epsilon/4$-differential privacy. Again, by the Composition
  Theorem in \cite{DMNS06}, these privacy parameters compose so that
  the strategy profile selection as a public message satisfies
  $\eps$-differential privacy.
  The action of each player is a function
  of only the public message and her own private payoff data (type), so
  by Lemma \ref{lem:billboard}, \Cref{alg:pselection} satisfies
  $\eps$-joint differential privacy.


We now prove that \Cref{alg:pselection} computes an approximate
equilibrium with quality score close to $\OPT$.  We know that with
probability at least $1-\beta$, all four instantiations of $\SV$ have
error at most
\[
err = \frac{100\gamma(\log(2Wn) + \log(8/\beta))}{\eps} \leq \alpha,
\]
by our assumption on $\alpha$. 
For the rest of the argument, we assume this level of accuracy, which
is the case except with probability $\beta$. 

Suppose that the algorithm outputs a strategy profile $\vec{y}$ from the
first two instantiations of $\SV$. Let $m$ be the index of the
corresponding query, then
\[
|S(\vec{y}) - s_m| \leq 3\alpha + err \leq 4\alpha.
\]
Since each player is playing a $(2\alpha + \gamma +
\zeta)$-aggregative best response to $s_m$ in the profile $\vec{y}$,
we know by Lemma~\ref{lem:abr-move} that each player is at least
playing a $(10\alpha + \gamma +\zeta)$-aggregative best response to
$S(\vec{y})$, and so by Lemma~\ref{lem:abr-prop}, a $(10\alpha + 2
\gamma + \zeta)$-best response.

Suppose that the algorithm does not output anything in the first two
instantiations of $\SV$. By the accuracy guarantee, we know that for
each $s_i\in \{s_1, \ldots , s_{2W/\alpha}\}$,
\[
|S_{max}(s_i) - s_i| \geq  2\alpha \mbox{ and } |S_{min}(s_i) - s_i|\geq 2\alpha.
\]
Furthermore, we know there must exist a $\zeta$-approximate pure
strategy Nash equilibrium $\vec{x}^*$, so each player in $\vec{x}^*$
must be playing a $(\zeta + \gamma)$-aggregative best response to
$S(\vec{x}^*)$ by Lemma~\ref{lem:nash-prop}. Then by
Lemma~\ref{lem:abr-move}, the players are playing a $(2\alpha + \gamma
+ \zeta)$-aggregative best response for some $s_l\in Z$. For such an
$s_l$, it must be the case that $|S(\vec{x}^*) - s_l| \leq \alpha$ and
$\vec{x}^*\in X(s_l)$, so,
\[
S_{min}(s_l) < s_l < S_{max}(s_l),
\]
otherwise the first two instantiations of $\SV$ would have output
$x_{max}(s_l)$ or $x_{min}(s_l)$.  The third instantiation of $\SV$
would find us such an $s_l$. Now as with~\Cref{thm:psn-sum}, we can
find a strategy profile $\vec{z}$ with the last instantiation of $\SV$
such that $\vec{z}$ is between $x_{max}(s_l)$ and $x_{min}(s_l)$ and,
\[
|S(\vec{z}) - s_l| \leq \alpha + \gamma/2.
\]
By Lemma~\ref{lem:abr-prop} and Lemma~\ref{lem:abr-move}, each player
in $\vec{z}$ is playing a $(4\alpha + 2\gamma + \zeta)$-aggregative
best response to $s_l$, and so a $(4\alpha + 3 \gamma + \zeta)$-best
response.

Let $\vec{y}'$ be the equilibrium in $\cE(\zeta)$ that gives quality
$q(S(\vec{y}')) = \OPT$.  Such an optimal strategy profile $\vec{y}'$
would be among the strategy profiles that our algorithm searches.
Either we output the profile $\vec{y}'$, or we found a different
profile $\vx$ associated with a discretized aggregator of higher
quality score (because we enumerate aggregators with higher quality
$q$ first). Let $s'$ be the associated aggregator to our output
profile $\vx$.

 Note that $|S(\vec{y}')-s''| \leq \alpha$ for some discretized
 aggregator $s'' \in Z$ and also $|S(\vx) - s'| \leq 4\alpha$. Because
 the order in which we iterates the aggregators gives priority to
 higher quality, we know that $q(s') \geq q(s'')$, so $q(S(\vx)) \geq
 \OPT - 5\alpha \lambda$.
\end{proof}